\documentclass[lettersize,journal]{IEEEtran}
\usepackage{graphicx}%
\usepackage{multirow}%

\usepackage{amsthm}%
\usepackage{amsmath}
\usepackage{xcolor}%
\usepackage{textcomp}%
\usepackage{algorithmicx}
\usepackage{algorithm}%
\usepackage{algpseudocode}%

\usepackage{tikz}
\usepackage{url}
\usepackage{subfigure}
\usepackage[normalem]{ulem}
\usepackage{fancyhdr}

\newtheorem{theorem}{Theorem}[section]

\newtheorem{definition}{Definition}

\newtheorem{example}{Example}
\newtheorem{property}{Property}[section]

\begin{document}

\title{Efficient discovery of unique column combinations on disk-resident data with limited memory}

\author{Xiaolong~Wan, Xixian~Han
\thanks{The authors are with School of Computer Science and Technology, Harbin Institute of Technology, China. (e-mail: wxl@hit.edu.cn, hanxx@hit.edu.cn)}
\thanks{}}

\markboth{}%
{Shell \MakeLowercase{\textit{et al.}}: A Sample Article Using IEEEtran.cls for IEEE Journals}


\maketitle

\begin{abstract}
The discovery of unique column combinations (UCCs) is a core task in data profiling, describing the key constraints of a table. The existing algorithms cannot deal with large-scale disk-resident data well due to high memory consumption and computational cost. In this paper, a novel DUD algorithm is developed to efficiently discover UCCs on disk-resident data with limited memory, which is inspired by the relationship between UCC discovery and transversal hypergraph. Rather than complete difference set generation of quadratic complexity, DUD only generates partial difference sets for hypergraph construction, followed by minimal hitting set enumeration to generate candidates and a validation process. DUD devises a strategy to generate full useful difference sets by pairwise comparisons of tuples having the same values with respect to some selected attributes. A novel theorem is developed and proved in this paper to report the candidates including the selected attributes as true UCCs directly without validation, which reduces the number of candidates to be validated significantly. A hash-based batch validation strategy is devised to validate a set of candidates on the relation instance, which only needs to maintain a small number of tuples in memory at a time. The extensive experimental results, conducted on synthetic and real-life data sets, show that DUD can discover UCCs on disk-resident data with high efficiency and low memory consumption.
\end{abstract}

\begin{IEEEkeywords}
	Unique column combination discovery, Disk-resident data, Limited memory, Hypergraph transversal.
\end{IEEEkeywords}

\section{Introduction} \label{sec:introduction}

The metaphors ``data is the new oil'' and ``analytics is the combustion engine'' \cite{DBLP:books/sp/fourthIR} assert that data has become a valuable resource for deriving insights that power modern industries. Before fully exploiting data, the necessary condition is to discover its metadata, which helps understand and process the data to make it truly an asset for decision-making. Data profiling is the systematic process of determining the metadata of data sets, which produces a small but informative summary of a given data set \cite{DBLP:series/synthesis/2018Abedjan}. 

The discovery of unique column combinations (UCCs) is a core task in data profiling, describing the key constraints of a table. Given relation schema \textit{R}, a UCC \textit{X} refers to a set of attributes whose projection only contains unique values in a given relation instance \textit{r}, i.e., any two distinct tuples $t_1, t_2 \in r$, $t_1[X] \neq t_2[X]$. Since UCCs identify tuples uniquely, they can serve as candidate keys, from which the primary key can be selected. Besides key discovery, UCCs also have many important practical cases including query optimization \cite{DBLP:journals/vldb/KossmannPN22}, data cleansing \cite{DBLP:books/acm/IlyasC19}, data integration \cite{DBLP:books/daglib/0029346}, database reverse engineering \cite{DBLP:journals/sigmod/Naumann13}. Due to various reasons, for example, undefined, missing or outdated, the UCC information is often not annotated in the data set, which motivates the automatic UCC discovery problem. Since any superset of a UCC is also a UCC, an important property for UCC is its minimality. A UCC \textit{X} is called minimal if none of its proper subsets is a UCC. The UCC discovery problem is to determine all minimal UCCs of a relation instance, which are sufficient to infer all UCCs. 

\begin{table}[t]
	\centering
	\caption{PART data records.}
	\scriptsize
	\begin{tabular}{cccccc}
		\hline
		PKEY & NAME & MFGR & BRAND & TYPE & PRICE \\
		\hline
		1 & Golden & M1 & B13 & PROMO & 901 \\
		2 & Blush & M1 & B13 & LARGE & 902\\
		3 & Spring & M4 & B42 & STAND & 903 \\
		4 & Cornfl & M3 & B34 & SMALL & 904 \\
		5 & Forest & M3 & B32 & STAND & 905 \\
		6 & Bisque & M2 & B24 & PROMO & 906 \\
		7 & Moccas & M1 & B11 & SMALL & 907 \\
		8 & Misty & M4 & B44 & PROMO & 908 \\
		9 & Thist & M4 & B43 & SMALL & 909 \\
		10 & Linen & M5 & B54 & LARGE & 910 \\
		\hline
	\end{tabular}
	\label{table:stuRecords}
\end{table}

\begin{example}
	Given a part table (Table \ref{table:stuRecords}) of a manufactory, each tuple keeps the information of a part, including part key (PKEY), part name (NAME), manufacturer (MFGR), brand (BRAND), type (TYPE) and retail price (PRICE). For the given relation instance, the discovered minimal UCCs are $\{$PKEY$\}$, $\{$NAME$\}$, $\{$PRICE$\}$, $\{$MFGR, TYPE$\}$, $\{$BRAND, TYPE$\}$.
\end{example}

Because of its practical importance, the UCC discovery problem has attracted extensive attention from researchers in both academia and industry \cite{DBLP:conf/cikm/AbedjanN11,DBLP:journals/pvldb/BirnickBFNPS20,DBLP:journals/pvldb/HeiseQAJN13,DBLP:conf/btw/PapenbrockN17,DBLP:conf/vldb/SismanisBHR06}. Actually, the discovery of UCC is a rather hard problem since the number of total UCC candidates is exponential to the number of attributes. The existing algorithms can be divided into three classes: attribute-based approaches, tuple-based approaches and hybrid approaches, which optimize the discovery process in different aspects. Optimizing for data sets with many tuples, attribute-based approaches \cite{DBLP:conf/cikm/AbedjanN11,DBLP:journals/pvldb/HeiseQAJN13} validate the candidates one after another, which are generated through attribute lattice. Aiming at the data sets with many attributes, tuple-based approaches \cite{DBLP:conf/vldb/SismanisBHR06} first compare the tuples pairwise and then derive all minimal UCCs by the comparison results. Hybrid approaches \cite{DBLP:journals/pvldb/BirnickBFNPS20,DBLP:conf/btw/PapenbrockN17} combine the advantages of the first two approaches and switch between two approaches for better performance, optimizing for data sets with many tuples and many attributes simultaneously.

The existing algorithms assume that the required data structures can be held entirely in memory. Their executions depend on some in-memory data structures, e.g. prefix tree or stripped partitions. However, the assumption fails in many real-life situations since the memory is limited for a specified machine and often is allocated partly to a process. On large-scale disk-resident data, the existing algorithms incur high execution cost, even cannot work when the required data structures cannot be maintained entirely in memory. \textit{Considering the ever-growing data volume, how to efficiently discover UCCs on disk-resident data under limited memory is still an important and open problem}. 

This paper solves UCC discovery problem inspired by the equivalency under parsimonious reduction between UCC discovery and transversal hypergraph \cite{DBLP:journals/tcs/BlasiusFS22}. Given relation instance, its UCCs are the minimal hitting sets of the hypergraph \cite{DBLP:journals/jcss/BlasiusFLMS22}, whose vertices are attributes of relation schema and the hyperedges are the minimal difference sets generated by the complete pairwise comparisons of tuples. For a pair of tuples, the \textit{difference set} is the set of attributes on which the two tuples have different values. A candidate attribute combination fails to be a UCC exactly when some pair of tuples agrees on every attribute in the candidate, that is, when the candidate shares no attribute with the difference set of that pair.

A novel DUD algorithm (\textit{D}isk-resident \textit{U}nique column combination \textit{D}iscovery) is proposed in this paper to efficiently discover all minimal UCCs on disk-resident data with limited memory. The execution of DUD consists of three phases. Since the cost of complete pairwise comparisons is prohibitively high on large-scale data, DUD first performs a subset of pairwise comparisons to build the hypergraph with partial information (phase 1). Next DUD discovers the UCC candidates by enumerating minimal hitting sets on the built hypergraph (phase 2). Then the candidates are checked by a validation process (phase 3). The candidates checked to be correct are valid UCCs, while the validation process of invalid candidates brings some \textit{conflicting tuples}, which have the same projection values on the invalid candidates. The hypergraph is updated by adding new hyperedges, i.e., the difference sets generated by comparing the conflicting tuples pairwise, and DUD enumerates minimal hitting sets again on the updated hypergraph. The following processing is similar, with the only difference that the newly generated candidates already discovered to be UCCs in the previous operations should be removed. If no new UCC candidates are generated or no new candidates are validated to be invalid, the execution of DUD is over. 

It is observed that the enumeration of minimal hitting sets is efficient in a hypergraph stemming from real-life data sets \cite{DBLP:journals/jcss/BlasiusFLMS22}. The emphasis of DUD lies on phase 1 and phase 3, i.e., efficient generation of partial difference sets and efficient validation of UCC candidates.

\textit{Emphasis 1: efficient generation of partial difference sets}. How to acquire a subset of difference sets, which can be generated with limited memory in a reasonable time and also be helpful for efficient UCC discovery, is a non-trivial issue. A difference set spanning all attributes provides no useful information for UCC discovery, since every non-empty candidate would trivially intersect it. The \textit{useful difference sets} for minimal hitting set enumeration, i.e., those that exclude at least one attribute, are generated by pairwise comparisons of tuples with at least one equal attribute value. An important theorem is proved in this paper that, \textit{if all useful difference sets with respect to an attribute are generated, any UCC candidates including the attribute are true UCCs definitely without validation}. This actually introduces a trade-off between partial difference set generation and validation process. This paper sets an upper-bound for the number of full useful pairwise comparisons with respect to each attribute in order to restrict the cost of generating difference sets, and the attributes are divided into two types: attributes of full difference set generation type (FG type) and attributes of no difference set generation type (NG type). DUD only generates full useful difference sets with respect to attributes of FG type, whose comparison numbers are no more than the upper-bound. Then, random sampling is performed on relation instance to supplement difference sets, whose pairwise comparison number is linear to tuple number.

\textit{Emphasis 2: efficient validation of UCC candidates}. The key to efficient validation with limited memory is not to maintain all tuples in memory. In this paper, a hash-based strategy is devised for validating the UCC candidates on the relation instance directly. Two tuples with the same projection on any candidate must have the same value for each constituent attribute. Therefore, the collision can be checked by first sorting the relation instance with respect to one constituent attribute and only maintaining the tuples with the same value of the constituent attribute in memory. The remaining \textit{validation candidates} contain only NG attributes. Therefore, DUD accesses NG attributes in the ascending order of maximum run lengths of equal attribute values to minimize the number of tuples kept in memory. A set of candidates sharing current attribute of NG type are validated simultaneously on the same sorted relation instance. The candidates checked to be correct are valid UCCs, while the validation process of invalid candidates brings some conflicting tuples, which are used to update the hypergraph.

The correctness of DUD is proved in this paper, along with its cost analysis. The extensive experiments are conducted on synthetic and real-life data sets, representing the practical data characteristics comprehensively. The experimental results show that DUD achieves remarkable efficiency on disk-resident data with limited memory.

The contributions of this paper are listed as follows:
\begin{itemize}
	\item[-] This paper proposes a novel DUD algorithm for efficient UCC discovery on disk-resident data.
	\item[-] The efficient generation of partial difference sets is devised in this paper, along with a novel theorem which prunes UCC candidates significantly.
	\item[-] The hash-based validation strategy is presented to validate candidates efficiently, by maintaining a small part of tuples in memory for collision checking.
	\item[-] The extensive experimental results show that DUD can discover UCCs on disk-resident data with high efficiency and low memory consumption.
\end{itemize}

The rest of the paper is organized as follows. Section \ref{sec:preliminary} gives preliminaries, followed by the related works in Section \ref{sec:relatedworks}. DUD algorithm is introduced in Section \ref{sec:dud}. Section \ref{sec:experiments} evaluates the performance of DUD and Section \ref{sec:conclusion} concludes the paper.

\section{Preliminaries} \label{sec:preliminary}

Given relation schema $R(A_1, A_2, \ldots, A_m)$, where $A_i (1 \le i \le m)$ is an attribute of \textit{R}, let \textit{r} be an instance of \textit{R} with \textit{n} tuples. Given $t \in r$, $\forall A \in R$, $t[A]$ represents the value of \textit{A} in \textit{t}, and $\forall X \subseteq R$, $t[X]$ denotes the projection of \textit{t} on \textit{X}. 

$\forall t_1, t_2 \in r$, if $t_1[X] = t_2[X]$, it means $t_1[A] = t_2[A]$ for $\forall A \in X$ and the projection of \textit{r} on \textit{X} contains duplicates. If $t_1[X] \neq t_2[X]$, it means $\exists A \in X, t_1[A] \neq t_2[A]$. Note that \textit{two tuples with different subscripts (e.g. $t_1, t_2 \in r$) correspond to two distinct tuples in this paper}. Given a subset \textit{X} of \textit{R}, we say \textit{X} is a unique column combination (UCC) if the projection of the relation instance \textit{r} on \textit{X} does not contain duplicates.

\begin{definition} \label{definition:ucc}
	Given relation instance \textit{r} of schema \textit{R}, an attribute subset \textit{X} ($X \subseteq R$) is a UCC, if and only if, $\forall t_1, t_2 \in r$, $t_1[X] \neq t_2[X]$.
\end{definition}

For a UCC \textit{X}, any superset of X is also a UCC. A UCC is called (inclusion-wise) minimal, if it does not contain any other UCC properly. 

\begin{definition} \label{definition:minucc}
	Given a UCC \textit{X}, it is minimal if and only if, $\forall Y \subset X$, \textit{Y} is a non-UCC, i.e., $\exists t_1, t_2 \in r$, $t_1[Y] = t_2[Y]$.
\end{definition}

Accordingly, the UCC discovery only returns minimal UCCs, and other UCCs can be derived naturally.

\begin{definition} \label{definition:uccdiscovery}
	Given relation instance \textit{r} of schema \textit{R}, the UCC discovery is to find all minimal UCCs in \textit{r}.
\end{definition}

\begin{figure}
	\centering
	\includegraphics[scale = 0.34]{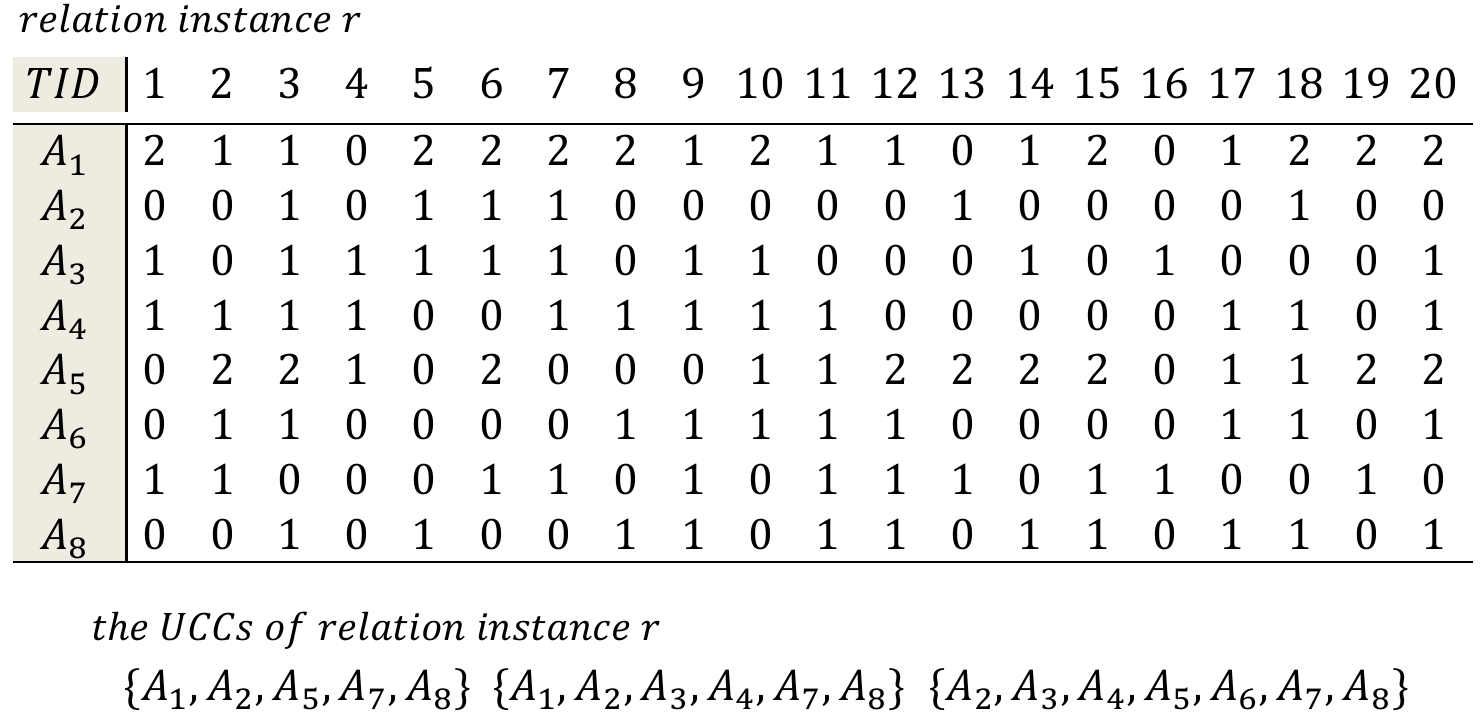}
	\caption{The illustration of relation instance in the running example.}
	\label{fig:RowTable}
\end{figure}

\begin{example}
	In this paper, we use a running example, as depicted in Figure \ref{fig:RowTable}, to illustrate the algorithm execution. The relation schema contains eight attributes $\{A_1, A_2, \ldots, A_8\}$. Figure \ref{fig:RowTable} also provides 3 UCCs of relation instance \textit{r}.
\end{example}

Below we introduce hypergraph and its intimate connection with UCC discovery. Many different data dependencies are discovered through enumerating hitting sets of hypergraphs \cite{DBLP:journals/pacmmod/BleifussPBSN24,DBLP:conf/icde/XiaoYTMW22}.

\textit{Hypergraph}. A hypergraph generalizes the concept of a graph, where each edge can connect multiple vertices rather than exactly two as in a conventional graph. Formally, a hypergraph $\mathcal{H}$ is an ordered pair $(V, E)$, where \textit{V} is a finite vertex set and \textit{E} is a hyperedge set. $\forall e \in E$, $e \in \mathcal{P}(V)$, where $\mathcal{P}(V)$ denotes the power set of \textit{V}. Given a hypergraph $\mathcal{H}$, a hitting set \textit{h} is a subset ($h \subseteq V$) of \textit{V} such that \textit{h} has at least one common element with every hyperedge, i.e., $\forall e \in E$, $h \cap e \neq \emptyset$. If \textit{h} does not include any other hitting set, it is called minimal. The minimal hitting sets of $\mathcal{H}$ form the transversal hypergraph $Tr(\mathcal{H})$. Given a hypergraph $\mathcal{H} = (V, E)$, its minimization $min(\mathcal{H})$ is a hypergraph formed by vertex set \textit{V} and inclusion-wise minimal hyperedges of \textit{E}. It is known that $Tr(\mathcal{H}) = Tr(min(\mathcal{H}))$ \cite{DBLP:journals/jcss/BlasiusFLMS22}.

$\forall t_1, t_2 \in r$, the difference set $DS(t_1, t_2)$ is the subset of \textit{R}, in which $t_1$ and $t_2$ have different attribute values, i.e., $DS(t_1, t_2) = \{A \in R \ | \ t_1[A] \neq t_2[A]\}$. As discussed in Section \ref{sec:introduction}, an attribute set $X$ is a UCC exactly when $X$ has at least one common attribute with $DS(t_1, t_2)$ for every pair of tuples $t_1, t_2 \in r$. It is already proved that, given a hypergraph $\mathcal{H}$ using attributes in \textit{R} as vertex set and minimal difference sets generated by all pairwise comparisons of tuples as hyperedge set, the UCC discovery problem is equivalent to transversal hypergraph problem and the minimal UCCs are equal to the minimal hitting sets in $\mathcal{H}$ \cite{DBLP:journals/tcs/BlasiusFS22}.

\section{Related works} \label{sec:relatedworks}

UCC discovery is one component of data profiling and metadata discovery \cite{DBLP:series/synthesis/2018Abedjan,DBLP:journals/sigmod/Naumann13}, a field concerned with inferring structural properties of datasets. Related problems include functional dependency (FD) discovery \cite{DBLP:journals/tkde/WanHWL24}, denial constraint discovery \cite{DBLP:journals/pvldb/ChuIP13}, and inclusion dependency discovery \cite{DBLP:journals/pvldb/PapenbrockKQN15}. Handling data that cannot be held entirely in memory is a common challenge across these problems. Scalable methods for FD discovery on massive data have also been developed \cite{DBLP:journals/tkde/WanHWL24}. Beyond relational tables, structural discovery has also been studied for semi-structured data, including schema inference for JSON \cite{DBLP:conf/edbt/BaaziziLCGS17} and FD discovery for XML \cite{DBLP:journals/vldb/YuJ08}. This paper focuses on UCC discovery, which finds the minimal column combinations that uniquely identify tuples in a relation. The existing algorithms can be classified into three categories: attribute-based approaches, tuple-based approaches, and hybrid approaches.

\textit{Attribute-based approaches}. The search space is modeled as an attribute lattice and each node corresponds to a UCC candidate. The candidates are enumerated in a systematic strategy and validated by stripped partitions \cite{DBLP:journals/cj/HuhtalaKPT99} to determine the true UCCs. HCA algorithm \cite{DBLP:conf/cikm/AbedjanN11} traverses the attribute lattice in bottom-top breadth-first strategy and generates the UCC candidates in a similar way as Apriori \cite{DBLP:conf/vldb/AgrawalS94}. The minimality pruning is used to reduce the search space. HCA validates the generated candidates by stripped partitions. Due to the relatively high validation cost, HCA utilizes data-based and statistics-based pruning to optimize the validation operation. DUCC algorithm \cite{DBLP:journals/pvldb/HeiseQAJN13} models the UCC discovery problem as a graph coloring problem. By traversing the attribute lattice in a depth-first and random walk, DUCC can quickly approach the border between UCCs and non-UCCs. The pruning strategy can be performed in a bottom-up and top-down manner at the same time. It is observed that the unreachable nodes may exist due to aggressive pruning, and a post-processing is presented to find and remove the unreachable parts. 

\textit{Tuple-based approaches}. All non-UCCs are discovered by first comparing pairs of tuples, and then the minimal UCCs are derived from the non-UCCs. GORDIAN algorithm \cite{DBLP:conf/vldb/SismanisBHR06} transforms the UCC discovery problem into cube computation problem, in which the count computation of the projected entities can be used to verify whether the projection is a UCC. The data set is first compressed into a prefix tree, and then the interleaving of cube computation and all non-UCC discovery is performed. The true UCCs can be computed by complementing the non-UCCs. The powerful pruning techniques are applied to reduce the time and space requirement, utilizing the idea that any subset of non-UCC is also a non-UCC.

\textit{Hybrid approaches}. The advantages of attribute-based approaches and tuple-based approaches are combined to achieve better performance for hybrid approaches. HyUCC algorithm \cite{DBLP:conf/btw/PapenbrockN17} automatically switches between attribute-based search (sampling phase) and tuple-based search (validation phase), depending on which one performs better currently. HyUCC starts with sampling phase, which performs sampling on the tuple pairs by stripped partitions to derive non-UCCs. The sampling efficiency is measured by the number of newly discovered non-UCCs per comparison. When the sampling efficiency falls below a threshold, sampling phase becomes inefficient, then the discovered non-UCCs are used to induce the candidate UCCs and HyUCC switches to validation phase. The validation efficiency is measured by the number of valid UCCs per validation. When the validation efficiency does not reach a threshold, the algorithm switches back to sampling phase with new comparison suggestions. HyUCC relaxes the threshold a bit in every switch and ends with validation phase. HPIValid algorithm \cite{DBLP:journals/pvldb/BirnickBFNPS20} models the UCC discovery problem as hitting set enumeration problem in hypergraph. Instead of comparing all pairs of tuples, HPIValid performs the UCC discovery with partial information only. The algorithm first samples the tuple pairs by stripped partitions to construct a partial hypergraph, on which the minimal hitting sets are enumerated. The newly discovered candidate UCCs are validated via stripped partitions to check whether they are true UCCs. For the invalid UCCs, new comparisons of tuple pairs are suggested to generate new hyperedges in the hypergraph. The algorithm continues until no new candidate UCCs are found by enumerating hitting sets in the updated hypergraph.

UCC discovery is closely related to the hypergraph transversal problem, which has been studied independently in combinatorics \cite{DBLP:journals/jcss/BlasiusFLMS22}. The enumeration of minimal hitting sets is equivalent to enumerating the minimal transversals of a hypergraph. DUD relies on the MMCS algorithm \cite{DBLP:journals/dam/MurakamiU14} for this enumeration step. Blasius et al. \cite{DBLP:journals/jcss/BlasiusFLMS22} show that hypergraphs arising from real-life data tend to have compact structure, making enumeration tractable in practice. Similar formulations have been applied in the discovery of other data dependencies \cite{DBLP:journals/pacmmod/BleifussPBSN24,DBLP:conf/icde/XiaoYTMW22}.

\textit{Discussion}. The existing UCC discovery algorithms assume that the memory is large enough to hold the required data structures entirely. The assumption often fails in practical situations. Prefix tree utilized in GORDIAN algorithm often occupies a larger size than relation instance \cite{DBLP:journals/datamine/HanPYM04} and cannot be held in memory entirely on data set of large scale. The stripped partitions used for attribute-based and hybrid approaches have nearly the same size as relation instance in the compressed representation \cite{DBLP:conf/sigmod/PapenbrockN16,DBLP:journals/tkde/WanHWL24}, and it is rather difficult for stripped partition based algorithms to discover UCCs when the size of stripped partition is larger than the allocated memory. 

The algorithm proposed in this paper follows the idea of hypergraph-based UCC discovery, placing it among the hybrid approaches discussed above. However, unlike these approaches, DUD does not rely on a global per-tuple structure, such as a prefix tree or stripped partitions, whose size scales with the relation instance. Instead DUD operates directly on the tuples of the relation instance throughout all phases. It is different from the existing hybrid approaches essentially in terms of two aspects: partial difference set generation strategy and validation strategy, which also are two most important components in hybrid approaches. 
\begin{itemize}
	\item[-] Instead of generating difference sets by sampled tuple pairs, our algorithm generates full useful difference sets with respect to some attributes. By the generation strategy, our algorithm can report most of the UCC candidates, which satisfy the specified condition, as true UCCs without validation. In contrast, the existing algorithms lack a similar pruning strategy and have to validate all candidates. 
	\item[-] The existing algorithms require the relatively expensive validation for the generated candidates one by one through manipulating stripped partitions. Our algorithm utilizes a hash-based batch validation strategy to validate a set of candidates on relation instance simultaneously by maintaining a small number of tuples in memory at a time.
\end{itemize}
This paper develops the algorithm elaborately, considering both high execution efficiency and the memory constraint throughout the algorithm design.

\section{DUD algorithm} \label{sec:dud}

This section introduces DUD algorithm (\textit{D}isk-resident \textit{U}nique column combination \textit{D}iscovery), which efficiently discovers UCCs on disk-resident data with limited memory without any pre-processing.

\textit{Algorithm overview}. DUD utilizes partial information, i.e., difference sets generated by a subset of tuple pairs, to discover UCCs by enumerating minimal hitting sets on the built hypergraph and validating the generated UCC candidates. In the beginning, DUD builds the \textit{full useful difference sets} with respect to some selected attributes. This provides a significant advantage that DUD does not need to validate each candidate and most candidates can be reported to be UCCs directly. For the validation candidates, DUD develops a hash-based batch validation strategy. By the validation process, DUD refines the built hypergraph if necessary and continues the UCC discovery until all UCCs are reported. 

Specifically, the execution of DUD consists of three phases, whose pseudocode is listed in Algorithm \ref{alg:dud}. In phase 1, DUD computes the difference sets by comparing tuples pairwise and constructs the hypergraph $\mathcal{H}$ (Line \ref{alg1:line:1} in Algorithm \ref{alg:dud}). Note that we do not compare all pairs of tuples in phase 1 but generate the difference sets by comparing a subset of tuple pairs. Then, the minimal hitting sets of $\mathcal{H}$ are enumerated in phase 2 (Line \ref{alg1:line:2} in Algorithm \ref{alg:dud}), which are the candidate UCCs. If the candidates need to be validated, a hash-based validation process is executed in phase 3 without maintaining all tuples in memory (Line \ref{alg1:line:5} in Algorithm \ref{alg:dud}). If all candidates are proved to be true UCCs, the execution is over. Any candidates proved to be non-UCCs indicate some conflicting tuples, which have duplicate projections on the candidates. The hypergraph $\mathcal{H}$ is updated by adding new hyperedges, i.e., newly generated difference sets by comparing the conflicting tuples pairwise (Line \ref{alg1:line:7} in Algorithm \ref{alg:dud}). In this case, the execution returns to phase 2 and enumerates minimal hitting sets again (Line \ref{alg1:line:8} in Algorithm \ref{alg:dud}). The following operation is similar, with the only difference that the newly generated candidates already discovered to be UCCs in the previous operations should be removed (Line \ref{alg1:line:4} in Algorithm \ref{alg:dud}). If no new UCC candidates are generated in phase 2 or no new candidates are validated to be invalid in phase 3, the execution of DUD is over.

\begin{algorithm}[t]
	\renewcommand{\algorithmicrequire}{\textbf{Input:}}
	\renewcommand{\algorithmicensure}{\textbf{Output:}}
	\renewcommand{\algorithmiccomment}[1]{ #1}
	
	\caption{DUD algorithm}
	\label{alg:dud}
	
	\footnotesize
	\begin{algorithmic}[1]
		\Require relation instance \textit{r} of \textit{m} attributes 
		\Ensure The discovered UCCs		
		
		\State $ST_{ucc}$: an empty set to keep the true UCCs
		
		\noindent\Comment{//Phase 1: generate initial difference sets and construct the hypergraph}
		\State Hypergraph $\mathcal{H}$ = BuildInitialHypergraph(\textit{r}) \label{alg1:line:1}
		
		\noindent\Comment{//Phase 2: Call MMCS \cite{DBLP:journals/dam/MurakamiU14} to find minimal hitting set}
		\State The UCC candidate set $ST_c$ = MMCS($\mathcal{H}$)\label{alg1:line:2}
		
		\noindent\Comment{//Check whether new candidates are generated.}
		\State \textbf{while} ($ST_c = ST_c \setminus ST_{ucc}$ is not empty) \textbf{do}\label{alg1:line:4}
		
		\noindent\Comment{//Phase 3: validate candidates by hashing, $\Delta_{DS}$ is a set to maintain difference sets generated by conflicting tuples}
		\State \quad $\Delta_{DS}$ = ValidateByHash($ST_c$, $ST_{ucc}$) \label{alg1:line:5}
		
		\noindent\Comment{//Generate new difference sets, enumerate minimal hitting sets on updated $\mathcal{H}$}
		\State \quad if($|\Delta_{DS}|$ $\neq$ 0) \textbf{then}\label{alg1:line:6}
		\State \quad\quad Update $\mathcal{H}$ by adding $\Delta_{DS}$ to hyperedge sets\label{alg1:line:7}
		\State \quad\quad $ST_c$ = MMCS($\mathcal{H}$)\label{alg1:line:8}
		\State \quad \textbf{else}
		\State \quad\quad $ST_c = \emptyset$
		\State \textbf{end while}
		\State \textbf{return} $ST_{ucc}$ \label{alg1:line12}
	\end{algorithmic}
\end{algorithm}

\subsection{Phase 1: construction of hypergraph}\label{sec:dud:phase1}

Phase 1 computes the useful difference sets with a relatively low cost for disk-resident relation instance. A naive choice that compares all pairs of tuples is impractical on large-scale data. Given a relation instance of \textit{n} tuples, the number of tuple pairs is $\binom{n}{2}$. It takes a very long time to execute $\binom{n}{2}$ comparisons when \textit{n} is large. The more ideal choice is to perform partial comparisons. The emphasis of phase 1 is how to compare a subset of tuple pairs with a rational cost to generate the difference sets, which can help discover UCCs quickly. 

We first introduce  Property \ref{property:alldifferenceset} that the difference set consisting of all attributes of relation schema \textit{R} does not affect the enumeration of minimal hitting sets if there is any difference set which is a nonvoid proper subset of \textit{R}. 

\begin{property} \label{property:alldifferenceset}
	The difference set $DS_F$ consisting of all attributes does not affect the enumeration of minimal hitting sets if any difference set $DS_{sub}$ which is a nonvoid proper subset of \textit{R} exists.
\end{property} 
\begin{proof}
	Given difference set $DS_F$, it is obvious that difference set $DS_{sub}$ is a subset of $DS_F$. As long as $DS_{sub}$ exists, the minimal hitting sets remains unchanged even if $DS_F$ is removed.
\end{proof}

According to Property \ref{property:alldifferenceset}, DUD only needs to generate the \textit{useful difference sets}, namely those that exclude at least one attribute, which requires the two tuples generating them to have at least one equal attribute value.

\textit{Full generation of useful difference sets with respect to attribute $A_i$}. Given instance \textit{r} of relation schema $R(A_1, A_2, \ldots, A_m)$, the difference sets with respect to $A_i$ ($1 \le i \le m$) are generated by pairs of tuples with equal $A_i$ values. Generating difference sets via a nested-loop comparison of all tuple pairs would take $O(n^2)$ time. DUD avoids this by first sorting the relation instance \textit{r} on the specified attribute, so that tuples with equal values on the attribute are grouped together and only tuples within the same run need to be compared. In this sense, sorting serves as an efficient primitive for reducing pairwise comparisons \cite{Roughgarden_2022}. Two-pass multi-way merge sort (TPMMS) \cite{DBLP:books/daglib/0020812} is used to sort files which exceed the sorting buffer size, while files within this size are loaded into memory, sorted, and written back to disk without intermediate temporary files. Let $r_i$ be the sorted version of \textit{r} in ascending order of $A_i$, and $|A_i|$ be the cardinality of $A_i$. The tuples in $r_i$ with the same $A_i$ attribute values are arranged consecutively and compared pairwise. Let $n_{i,k} (1 \le k \le |A_i|)$ be the number of tuples with attribute value $v_{i,k}$ of $A_i$, where $v_{i,1} \le v_{i,2} \le \ldots \le v_{i,|A_i|}$ and $n = \sum_{k = 1}^{|A_i|} n_{i,k}$. Hence, the pairwise comparison number $cmp_i$ performed in $r_i$ is $\sum_{k=1}^{|A_i|} \binom{n_{i,k}}{2}$, which usually is much lower than $\binom{n}{2}$. When computing the difference sets for $r_i$, DUD does not need to keep all tuples in memory, but processes the tuples with $A_i = v_{i,k} (1 \le k \le |A_i|)$ individually. The maintained tuples are compared pairwise and the generated difference sets are stored in set \textit{MDS} to avoid duplication.

However, there still remains one potential problem. If the cardinality of $A_i$ is low or there exists many duplicates for some attribute values, the value of $cmp_i$ can be very large. In an extreme case, $|A_i| = 1$, $\binom{n}{2}$ pairwise comparisons have to be executed. For this reason, DUD does not generate full difference sets with respect to the attributes whose comparison numbers are too large. In this paper, a parameter $\gamma$ is specified as an upper-bound of comparison numbers. If $cmp_i \le \gamma$, $A_i$ is called attribute \textit{of full difference set generation type} (abbreviated as \textit{attribute of FG type}), and $cmp_i$ comparisons are performed to generate difference sets. Otherwise, if $cmp_i > \gamma$, $A_i$ is called attribute \textit{of no difference set generation type} (abbreviated as \textit{attribute of NG type}), and DUD \textit{does not} generate difference sets with respect to $A_i$ specifically.

\textit{Determination of attributes type}. If $A_i$ is of NG type, there is no need to sort \textit{r} with respect to $A_i$ at the very start actually. We need a determination process with a much lower cost to determine attribute types (FG type or NG type). The number of $cmp_i$ is concerned only with the $A_i$ values. Before sorting \textit{r} and generating difference sets, DUD first decomposes \textit{r} as \textit{m} column files $L_1(A_1), L_2(A_2), \ldots, L_m(A_m)$ \cite{DBLP:conf/vldb/StonebrakerABCCFLLMOORTZ05}. Then $L_i$ ($1 \le i \le m$) is sorted and the number of $cmp_i$ is computed by counting the occurrence number of each distinct value. Note that DUD does not maintain all values of $L_i$ in memory but only keeps a counter for each run to accumulate the value of $cmp_i$. If $cmp_i \le \gamma$, $A_i$ is an attribute of FG type. Otherwise, $A_i$ is an attribute of NG type.

In this paper, let $S_{FG}$ be the set of attributes of FG type, and $S_{NG}$ be the set of attributes of NG type.

\textit{Determination of parameter $\gamma$}. For DUD, we must determine the parameter $\gamma$ for its execution. Given $r_i$, the maximum number of possible comparisons is $\binom{n}{2}$, which is too large obviously. In this paper, $\gamma$ is specified to $3 \times n \times \log_2 n$ by default, which is smaller than $\binom{n}{2}$ substantially and still large enough to generate useful difference sets. The rationality of determination of $\gamma$ is verified in the experiments.

\textit{Supplement of difference sets}. The attributes of NG type, whose useful comparison numbers are larger than $\gamma$, are skipped in difference set generation. After all useful difference sets are generated with respect to attributes of FG type, DUD performs a uniform sampling operation on \textit{r}. Let $|S_{NG}|$ be the number of attributes of NG type, the total pairwise comparison number for the sampled tuples is set to $n \times |S_{NG}|$, and the number of sampled tuples is $\sqrt{n \times |S_{NG}|}$. The aim of the above operations is to supplement the difference sets due to the neglect of attributes of NG type with a low enough cost and small enough memory consumption.

In phase 1, DUD generates useful difference sets and maintains their minimal form in \textit{MDS}. Generated difference sets are first collected in a hash set using their bit-set representation, which removes duplicate difference sets. DUD does not insert them into \textit{MDS} one by one. Instead, minimization is performed in batch: if a collected difference set properly contains another collected difference set, the larger one is removed, so that a difference set is kept only when no proper subset of it remains in the collected set. The initial minimization has worst-case cost $O(k^2)$, where $k$ is the number of distinct collected difference sets before minimization. When new difference sets are generated during validation, DUD applies a delta-based minimization procedure: it compares old difference sets with new ones and new difference sets with each other, but does not compare old difference sets among themselves again, since they have already been minimized in the previous iteration. In this way, DUD constructs a hypergraph $\mathcal{H}(V, E)$, whose vertex set is the attribute set of the relation schema and whose hyperedge set consists of the minimal difference sets in \textit{MDS}.

\begin{algorithm}[t]
	\renewcommand{\algorithmicrequire}{\textbf{Input:}}
	\renewcommand{\algorithmicensure}{\textbf{Output:}}
	\renewcommand{\algorithmiccomment}[1]{ #1}
	
	\caption{BuildInitialHypergraph}
	\label{alg:phase1}
	
	\footnotesize
	\begin{algorithmic}[1]
		\Require relation instance \textit{r} of \textit{m} attributes 
		\Ensure The initial hypergraph $\mathcal{H}$		
		
		\State Decompose $r(A_1, \ldots, A_m)$ into column files $L_1, \ldots, L_m$ \label{alg2:line:1}
		\Comment{//Compute the comparison number for each attribute}
		\State \textbf{for} ($i = 1; i \le m; i=i+1$) \label{alg2:line:2}
		\State \quad Sort $L_i$ with respect to $A_i$ and compute $cmp_i$ of $L_i$\label{alg2:line:3}
		
		
		\noindent\Comment{//Compute the value of $\gamma$}
		\State $\gamma = 3 \times n \times \log_2 n$ \label{alg2:line:4}
		
		\noindent\Comment{//Generate difference sets initially}
		\State \textbf{for} ($i = 1; i \le m; i=i+1$) \label{alg2:line:9}
		
		\noindent\Comment{//Current attribute is of FG type}
		\State \quad \textbf{if} $cmp_i \le \gamma$ \textbf{then}
		\State \quad\quad Sort \textit{r} to be $r_i$ with respect to $A_i$ and   				
		
		\Comment{generate difference sets by tuples with equal $A_i$ values} \label{alg2:line:11}
		
		\noindent\Comment{//Supplement of difference sets}
		\State Randomly sample $\sqrt{n \times |S_{NG}|}$ tuples\label{alg2:line:12} 
		\State Compare samples pairwise and build difference sets \label{alg2:line:13}		
		
		\noindent\Comment{//Build initial hypergraph by minimal difference sets}
		\State Minimize generated difference sets in set MDS \label{alg2:line:14}
		\State Build hypergraph $\mathcal{H}(V, E)$ by difference sets in MDS \label{alg2:line:15}
	\end{algorithmic}
\end{algorithm}

\textit{Execution process of phase 1}. To sum up, the execution process of phase 1 is described as below, whose pseudocode is listed in Algorithm \ref{alg:phase1}. In phase 1, DUD first decomposes the relation instance \textit{r} into \textit{m} column files (Line \ref{alg2:line:1} in Algorithm \ref{alg:phase1}), computes the useful comparison number for each column file (Line \ref{alg2:line:2} and Line \ref{alg2:line:3}), specifies $\gamma$ value (Line \ref{alg2:line:4}) for determining the attribute types. For each attribute $A_i$ of FG type, DUD sorts \textit{r} in ascending order of $A_i$ and performs the useful pairwise tuple comparisons to generate difference sets (Line \ref{alg2:line:9} to Line \ref{alg2:line:11}). After generating difference sets with respect to all attributes in $S_{FG}$, DUD executes a randomly uniform sampling on \textit{r} to supplement the difference sets (Line \ref{alg2:line:12} and Line \ref{alg2:line:13}). When phase 1 terminates, DUD builds the initial hypergraph with the generated difference sets as hyperedges (Line \ref{alg2:line:14} and Line \ref{alg2:line:15}).

On generating difference sets, it is found that the number of the generated difference sets is far smaller than the number of total pairwise comparisons. This is verified in the benchmark data sets and real-life data sets, which are utilized in the performance evaluation.

\begin{figure}
	\centering
	\includegraphics[scale = 0.32]{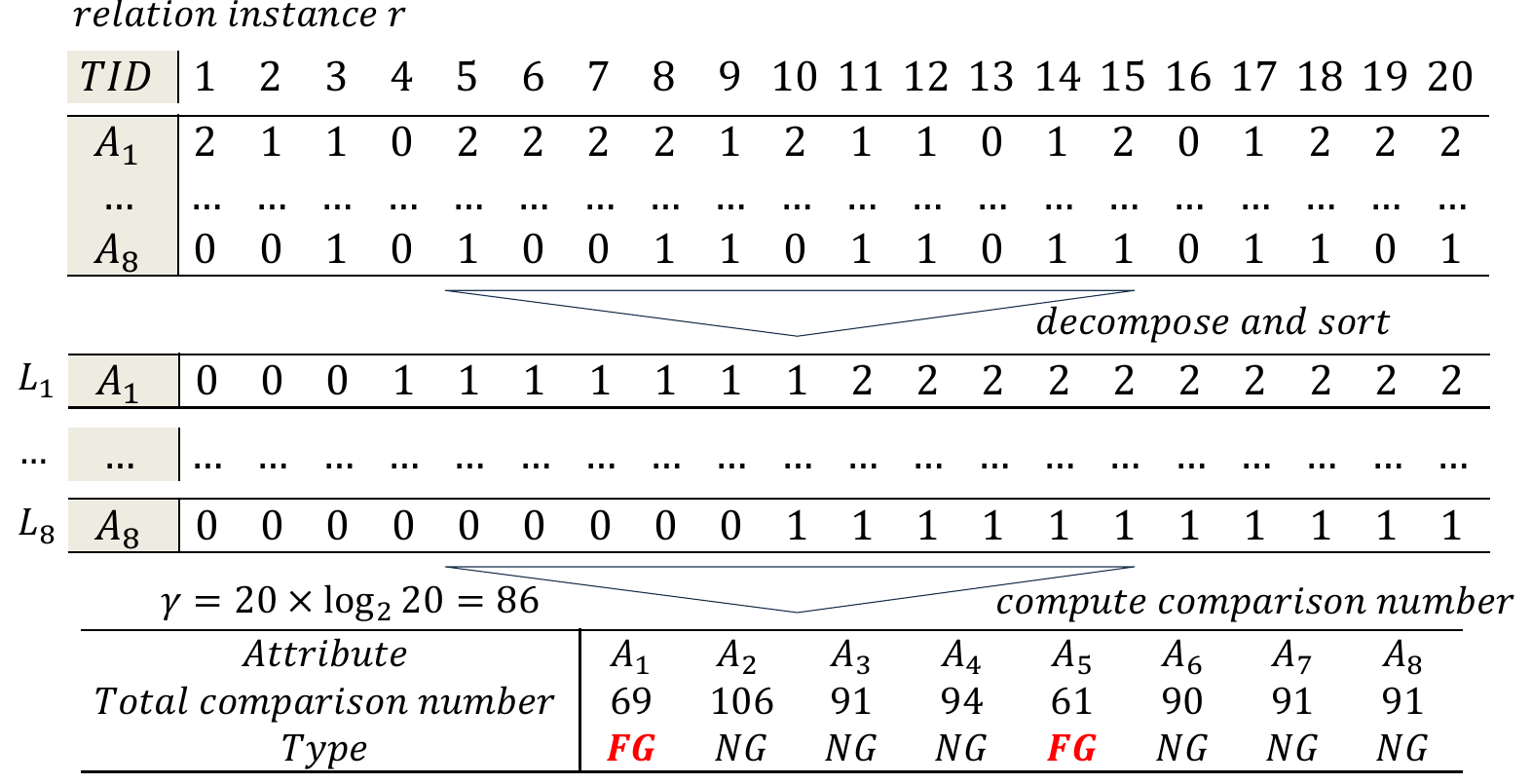}
	\caption{The illustration of determining attribute types.}
	\label{fig:determineattributetype}
\end{figure}

\begin{figure}
	\centering
	\includegraphics[scale = 0.315]{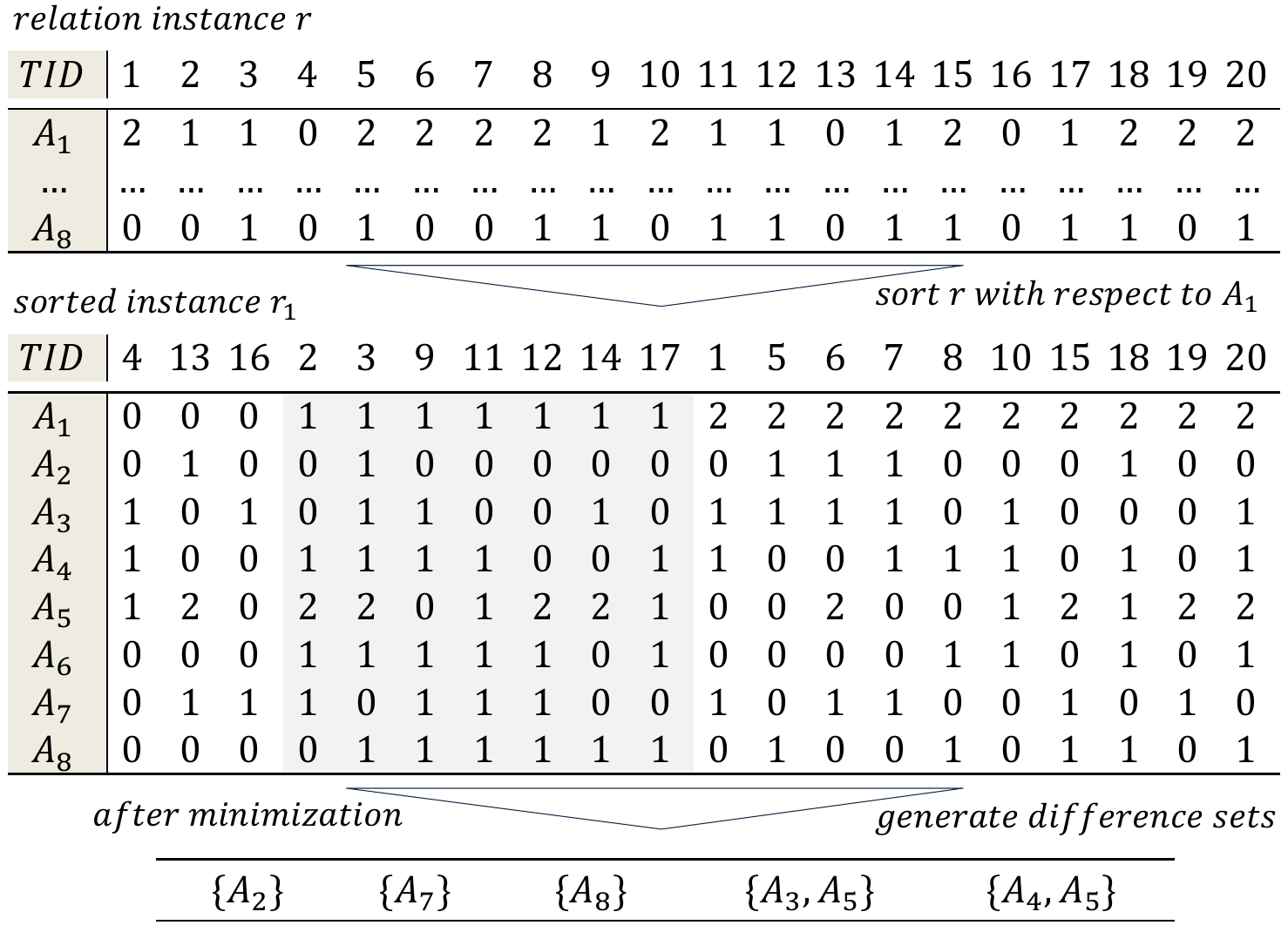}
	\caption{The illustration of generating difference sets for $r_1$.}
	\label{fig:ds1}
\end{figure}

\begin{figure}
	\centering
	\includegraphics[scale = 0.3]{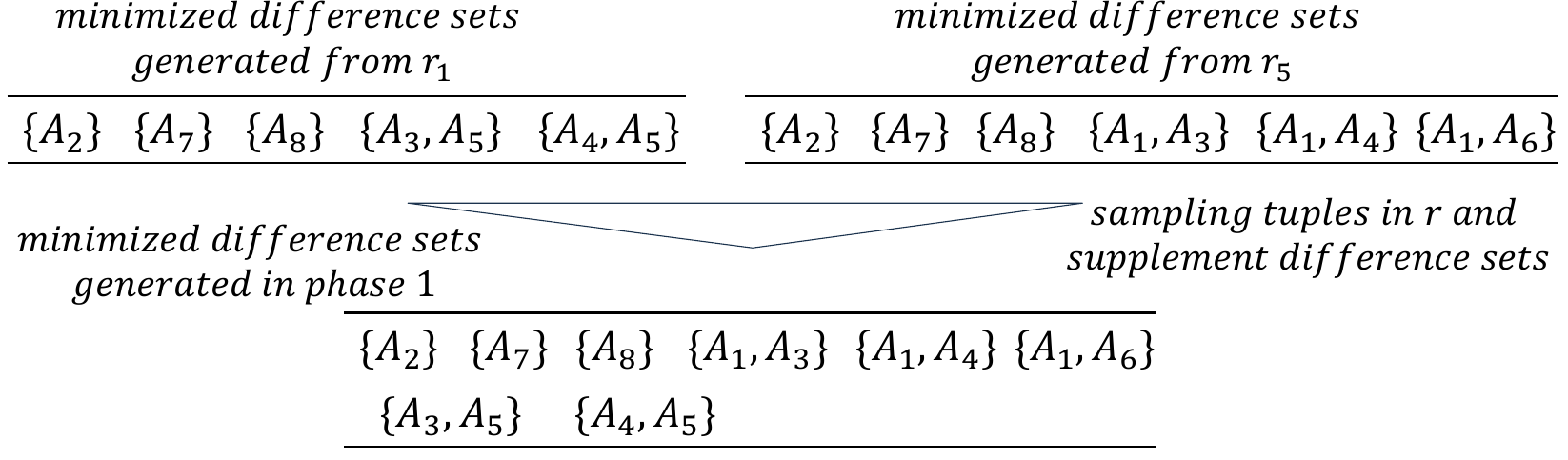}
	\caption{The illustration of difference sets in phase 1.}
	\label{fig:dsinphase1}
\end{figure}

\begin{example}
	The execution of phase 1 is illustrated in Figure \ref{fig:determineattributetype}, Figure \ref{fig:ds1} and Figure \ref{fig:dsinphase1}. The determination of attribute type is depicted in Figure \ref{fig:determineattributetype}. After decomposing and sorting column files, the comparison number for each attribute is computed. Then the value of $\gamma$ is set to 86\footnote{Due to too small scale, $\gamma$ is set to $n \times \log_2 n$ in running example provisionally rather than $3 \times n \times \log_2 n$ in this paper.} in the running example. For example, $cmp_1 = 69 \le 86$ and $A_1$ is of FG type, $cmp_2 = 106 > 86$ and $A_2$ is of NG type. Since $A_1$ is of FG type, its processing is shown in Figure \ref{fig:ds1} to generate difference sets. The relation instance \textit{r} is sorted with respect to $A_1$, and the difference sets are generated by pairwise comparisons of tuples with equal $A_1$ values in $r_1$. The difference sets generated for $r_1$ and $r_5$ are illustrated in Figure \ref{fig:dsinphase1}. With supplement of difference sets, DUD generates eight difference sets in total, which is used to build hypergraph for minimal hitting enumeration in phase 2. The built hypergraph is provided in Figure \ref{fig:phase2}.
\end{example}

\subsection{Phase 2: enumeration of minimal hitting sets}

By the built hypergraph $\mathcal{H}(V, E)$, DUD enumerates minimal hitting sets in phase 2, utilizing the existing algorithm MMCS \cite{DBLP:journals/dam/MurakamiU14} (Line \ref{alg1:line:2} of Algorithm \ref{alg:dud}). For better understanding and completeness of algorithm description, the implementation of MMCS is described below.

The enumeration of minimal hitting sets can be treated as breadth-first exploration on a set enumeration tree $\mathcal{T}$ \cite{DBLP:books/sp/fpm2014}, assuming that $A_1, A_2, \ldots, A_m$ already are arranged in lexicographic ordering. Every node in $\mathcal{T}$ has four fields: \textit{atts}, \textit{cand}, \textit{uncov} and \textit{crit}. 
\begin{itemize}
	\item[-] The first field \textit{atts} is the set of attributes maintained by the node, and $atts \subseteq V$.
	\item[-] The second field \textit{cand} is an array to maintain the attributes which can be added to \textit{atts} fields of its child nodes.
	\item[-] The third field \textit{uncov} is a set of hyperedges which have no common elements with \textit{atts}, i.e., $uncov = \{ e | e \in E, e \cap atts = \emptyset \}$.
	\item[-] The fourth field \textit{crit} keeps critical hyperedges for each attribute in \textit{atts}. $\forall A \in atts$, $crit(A) = \{ e | e \in E, e \cap att = A \}$. This can be implemented as a hash table mapping from single attribute to corresponding critical hyperedges.
\end{itemize}

Initially, the root \textit{rt} of $\mathcal{T}$ is created, with $rt.atts = \emptyset$, $rt.cand = \{ A_1, A_2, \ldots, A_m \}$, $rt.uncov = E$ and empty \textit{crit} field. Then \textit{rt} is added to a queue \textit{Q}. 

In each iteration, if \textit{Q} has no more element, phase 2 terminates. Otherwise, the head \textit{o} of \textit{Q} is removed and processed. If \textit{o} satisfies minimality condition, i.e., $o.uncov = \emptyset$ and $o.crit(A) \neq \emptyset$ for $\forall A\in o.atts$, it is proved that \textit{o} is a minimal hitting set \cite{DBLP:journals/dam/MurakamiU14}. Otherwise, if \textit{o} does not satisfy minimality condition, the child nodes of \textit{o} are constructed by adding each element in \textit{o.cand}. Let \textit{A} be the \textit{a}th element ($1 \le a \le |o.cand|$) in \textit{o.cand}, and $o.cand(a + 1, \ldots)$ be the elements in \textit{o.cand} after the \textit{a}th position. The \textit{a}th child node $cld_a$ of \textit{o} is generated with $cld_a.atts = o.atts \cup A$, $cld_a.cand = o.cand(a + 1, \ldots)$, $cld_a.uncov$ and $cld_a.crit$ can be updated by the information in \textit{o.uncov} and \textit{o.crit}. The child node $cld_a$ is added to \textit{Q} for the following iterations.

\textit{Pruning strategy}. Given current node \textit{o}, $\forall e \in o.uncov$, since a minimal hitting set should have at least one common element with every hyperedge, any minimal hitting set \textit{h} including $o.atts$ satisfies $h \cap e \neq \emptyset$. The possible extension to \textit{o.atts} should be limited to $e \cap o.cand$ rather than \textit{o.cand}. $\forall e \in o.uncov$, the hyperedge $e_m$ is chosen which minimizes $|e \cap o.cand|$. Therefore, only the attributes in $e_m$ are used to extend $o.atts$ and generate the child nodes of \textit{o}.

DUD maintains the generated minimal hitting sets (i.e., UCC candidates) of phase 2 in $ST_{c}$.

\begin{figure}
	\centering
	\includegraphics[scale = 0.3]{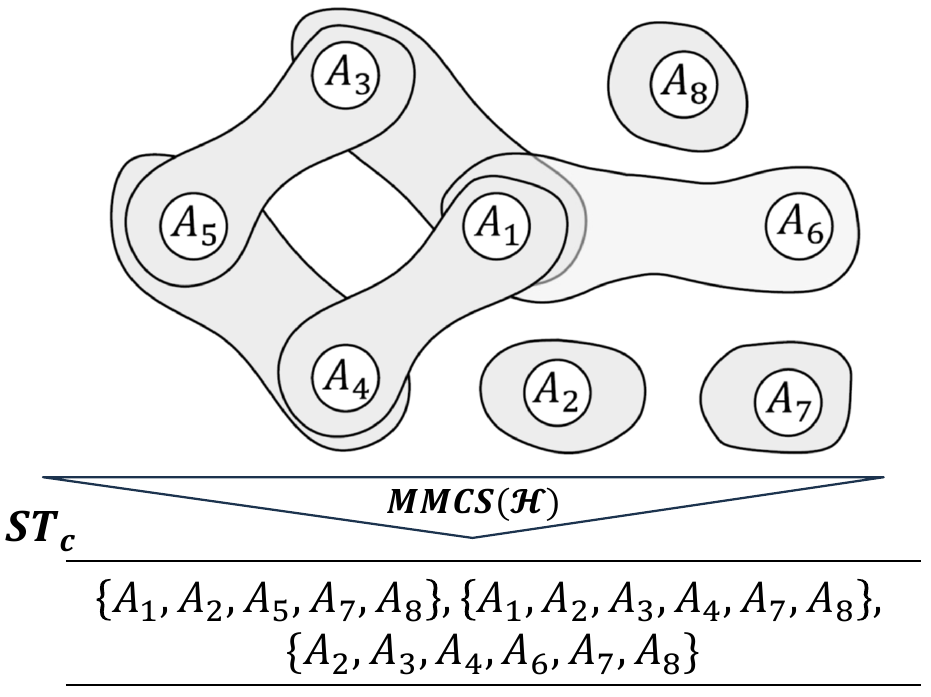}
	\caption{The illustration of execution results of phase 2 on hypergraph $\mathcal{H}$ built in phase 1, whose vertices are $\{ A_1, A_2, A_3, A_4, A_5, A_6, A_7, A_8 \}$ and hyperedges are $\{A_2\}$, $\{A_7\}$, $\{A_8\}$, $\{A_1, A_3\}$, $\{A_1, A_4\}$, $\{A_1, A_6\}$, $\{A_3, A_5\}$ and $\{A_4, A_5\}$.}
	\label{fig:phase2}
\end{figure}

\begin{example}
	Figure \ref{fig:phase2} depicts the results of phase 2 on the running example. Given the built hypergraph $\mathcal{H}$ in phase 1, phase 2 returns 3 UCC candidates, each of which has at least one common element with every hyperedge. 
\end{example}

\subsection{Phase 3: validation of UCC candidates}

Since phase 1 often does not generate complete difference sets, the UCC candidates obtained in phase 2 cannot be guaranteed to be true UCCs and need to be validated. The efficient validation involves two important issues. 
\begin{itemize}
	\item[-] \textit{Issue 1}. How to reduce the number of validated candidates in phase 3.
	\item[-] \textit{Issue 2}. How to reduce the validation cost for the candidates in phase 3 under limited memory.
\end{itemize}

For issue 1, a better choice is to report true UCCs in candidates without actual validation. An important statement is proved in Theorem \ref{theorem:candidatepruning} that helps reduce the number of validated candidates. It is proved that any candidates which contain at least one attribute of FG type are true UCCs definitely. Of course, these candidates do not need validation anymore and can be reported directly. 

\begin{theorem} \label{theorem:candidatepruning}
	Given a UCC candidate \textit{X}, if \textit{X} contains an attribute $A_i$ of FG type, it must be a UCC.	 
\end{theorem}
\begin{proof}
	$\forall t_1, t_2 \in r$, there exist two cases for the equality comparison between $t_1.A_i$ and $t_2.A_i$. If $t_1.A_i \neq t_2.A_i$, $t_1[X] \neq t_2[X]$ naturally since $A_i \in X$. If $t_1.A_i = t_2.A_i$, according to the execution in phase 1, the tuples with the equal $A_i$ attribute values are compared pairwise, whose generated difference sets do not include $A_i$. Since \textit{X} is a minimal hitting set generated in phase 2, it has at least one common element with every hyperedge. The common elements of \textit{X} with the hyperedge corresponding to difference set generated by comparing $t_1$ and $t_2$ do not include $A_i$ also. It means that $t_1$ and $t_2$ do not have equal value for at least one attribute in \textit{X} other than $A_i$, we have $t_1[X] \neq t_2[X]$. To sum up, $\forall t_1, t_2 \in r$, $t_1[X] \neq t_2[X]$ in any case, i.e., \textit{X} is a UCC.
\end{proof}

\begin{algorithm}[t]
	\renewcommand{\algorithmicrequire}{\textbf{Input:}}
	\renewcommand{\algorithmicensure}{\textbf{Output:}}
	\renewcommand{\algorithmiccomment}[1]{ #1}
	
	\caption{ValidateByHash}
	\label{alg:phase3}
	
	\footnotesize
	\begin{algorithmic}[1]
		\Require the candidate set $ST_c$ and true UCC set $ST_{ucc}$ 
		\Ensure a set $\Delta_{DS}$ keeping newly generated difference sets
		
		\State Initialize $\Delta_{DS}$ to be an empty set
		\State \textbf{foreach} $X \in ST_c$ \textbf{do}\label{alg3:line2}
		\State \quad \textbf{if} $X$ contains any attributes of FG type \textbf{then}
		\State \quad\quad Remove $X$ from $ST_c$ and put it to $ST_{ucc}$\label{alg3:line4}
		
		\State \textbf{if} $|ST_c| = 0$ \textbf{then}\label{alg3:line5}
		\State \quad \textbf{return} $\Delta_{DS}$\label{alg3:line6} \Comment{// Current $\Delta_{DS}$ is empty}
		
		\noindent\Comment{// Perform hash-based validation for candidates not directly reported as UCCs}
		
		\noindent\Comment{// The attribute of NG type in $S_{NG}$ are arranged in the ascending order of maximum run length}
		\State \textbf{foreach} $A_j \in S_{NG}$ \textbf{do} 
		
		\noindent\Comment{// $ST_{pis}$ keeps the pis of conflicting tuples}
		\State \quad Initialize an empty set $ST_{pis}$, whose elements are sets.
		\State \quad $C_j = \{ X | X \in ST_c \ and \ A_j \in X \}$ and $ST_c = ST_c \setminus C_j$\label{alg3:line9}
		
		\State \quad \textbf{if} ($|C_j| = 0$) \textbf{then}\label{alg3:line10}
		\State \quad\quad \textbf{continue}\label{alg3:line11}
		
		\State \quad Obtain $r_j$, the relation sorted by $A_j$; sort $r$ only if $r_j$ has not been materialized.
		
		\noindent\Comment{// Let $v_{j,1} \le \ldots \le v_{j,|A_j|}$ be distinct values of $A_j$}
		\State \quad \textbf{for} $(k = 1; k \le |A_j|; k = k + 1)$ \textbf{do}
		\State \quad\quad Initialize hash table $HT_X$ for $X \in C_j$\label{alg3:line14}
		\State \quad\quad Retrieve each tuple $t$ in $r_j$ with $A_j = v_{j,k}$\label{alg3:line15}
		\State \quad\quad Maintain $(t[X], t)$ in $HT_X$ for $X \in C_j$\label{alg3:line16}

		\State \quad\quad \textbf{foreach $X \in C_j$}
		
		\noindent\Comment{// \textit{key} is the projection on \textit{X} and \textit{value} is a set of tuples with the same projections.}
		\State \quad\quad\quad \textbf{foreach} $(key, value) \in HT_X$ \textbf{do}\label{alg3:line18}
		\State \quad\quad\quad\quad \textbf{if} $value.size \ge 2$ \textbf{then}\label{alg3:line19}
		\State \quad\quad\quad\quad\quad \textit{X} is an invalid UCC and add $value$ to $ST_{pis}$\label{alg3:line20}
		\State \quad\quad\quad\quad\quad Perform $\binom{|value|}{2}$ pairwise comparison, the  
		
		\Comment{\quad\quad\quad generated difference sets are inserted $\Delta_{DS}$}\label{alg3:line21}
		
		\State \quad Put candidates in $C_j$ which are not invalid to $ST_{ucc}$\label{alg3:line22}			
		
		\noindent\Comment{// All candidates are validated}
		\State \quad \textbf{if} ($|ST_c| = 0$) \textbf{then}\label{alg3:line23}
		\State \quad\quad \textbf{break}\label{alg3:line24}		
		
		\State \textbf{return} $\Delta_{DS}$			
	\end{algorithmic}
\end{algorithm}

\textit{Execution process of phase 3}. In phase 3, DUD first reports the true UCCs from candidates by Theorem \ref{theorem:candidatepruning}. This reduces the number of candidates to be validated. For the remaining validation candidates, the validation process is executed in several iterations. At each iteration, DUD sorts the relation instance with respect to a certain attribute of NG type,  and adopts hash-based validation by maintaining a subset of tuples with equal value of the attribute in memory at a time. Note that the validation is executed in batch mode, i.e., several candidates can be validated at the same time. This reduces the validation cost further. The candidates which are validated to be true UCCs are reported directly. Some candidates are validated to be non-UCCs because there exists conflict tuples, which are used to generate new difference sets and update the hypergraph for the following operations. The pseudocode of phase 3 is listed in Algorithm \ref{alg:phase3}.

Given the UCC candidates $ST_c$, $\forall X \in ST_{c}$, if $\exists A_i \in X$ and $A_i$ is of FG type, \textit{X} is reported to be a true UCC and removed from $ST_c$ (Line \ref{alg3:line2} to Line \ref{alg3:line4} in Algorithm \ref{alg:phase3}). In this paper, the discovered UCCs are stored in $ST_{ucc}$. In this way, the number of candidates to be validated can be reduced significantly based on the following \textit{intuition for candidate pruning}. At best, all candidates can be reported as true UCCs and the validation process of relatively high cost can be avoided entirely (Line \ref{alg3:line5} and Line \ref{alg3:line6}). 

\textit{Intuition for candidate pruning}. This intuition concerns the candidates initially generated in phase 2, before the pruning by Theorem \ref{theorem:candidatepruning} is applied. Many of these candidates may contain at least one FG attribute, and this is why the pruning can reduce the number of candidates to be validated. As stated in Section \ref{sec:dud:phase1}, the threshold $\gamma$ is used to distinguish FG attributes from NG attributes according to the number of pairwise comparisons required for generating useful difference sets. An FG attribute usually has relatively high cardinality, which leads to fewer tuple pairs with equal values on that attribute. Since UCCs are used to uniquely identify tuples, attributes with higher cardinalities are more likely to participate in UCCs. After this pruning step, the remaining validation candidates in $ST_c$ contain only NG attributes and are handled by hash-based batch validation.

\begin{figure}
	\centering
	\includegraphics[scale = 0.32]{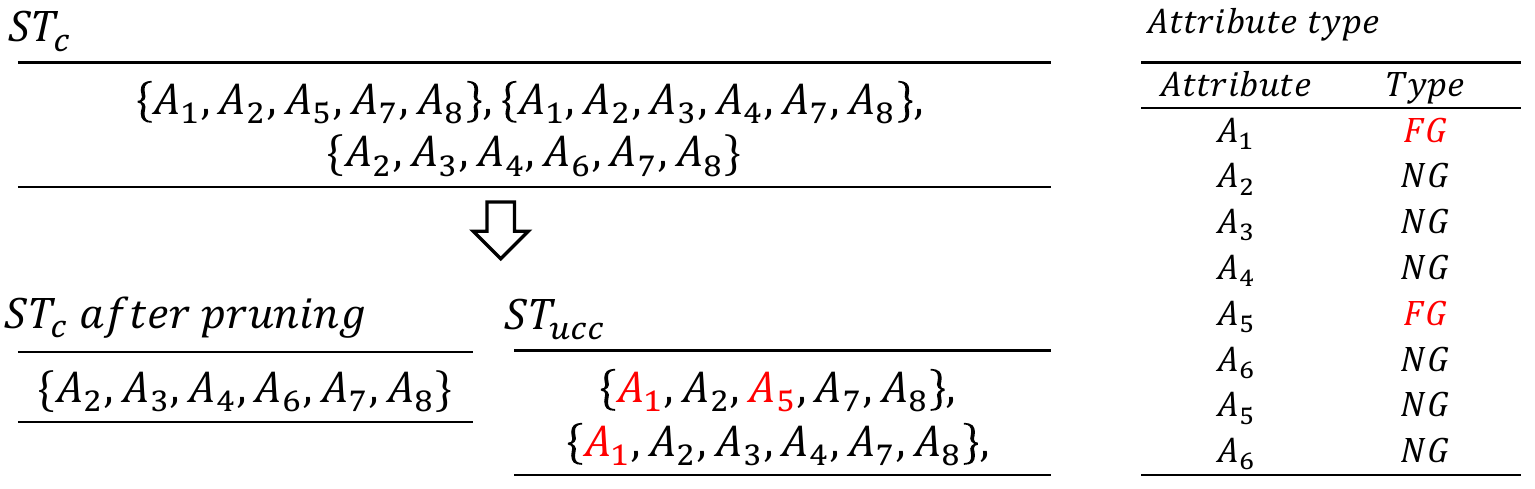}
	\caption{The illustration of directly reporting UCC candidates in the running example.}
	\label{fig:visio_candidatepruning}
\end{figure}

\begin{example}
	Figure \ref{fig:visio_candidatepruning} illustrates the effect of applying Theorem \ref{theorem:candidatepruning} to reduce the number of UCC candidates in the running example. For three UCC candidates generated in phase 2, there are two of them including $A_1$ or $A_5$, which are of FG type. These candidates can be proved to be true UCCs. Only one candidate is required for validation.
\end{example}

For the validation candidates in $ST_c$, the issue 2 comes up. This paper considers disk-resident data under limited memory, and the stripped partition based validation, which is commonly used in the existing algorithms, is inapplicable to DUD. Any data structures containing full data do not work under limited memory when they cannot be held entirely in memory. DUD utilizes hash-based strategy to perform validation on relation instance directly. The hash-based validation strategy is straightforward. Given any UCC candidate \textit{X}, any collision on a hash table with projection of tuples on \textit{X} as \textit{key} invalidates \textit{X}. \textit{The important aspect to consider here is how to achieve the collision checking without maintaining full data in memory}. Evidently, given a UCC candidate \textit{X}, any two tuples which have the same projection on \textit{X} must have the same value for each constituent attribute of \textit{X}. In this case, the collision can be checked by retrieving a relation instance, whose tuples are arranged in ascending order of one constituent attribute, and only maintaining the tuples with equal values of the attribute in memory. As illustrated in Figure \ref{fig:partialloading}, this reduces the required memory obviously. 

\begin{figure}
	\centering
	\renewcommand{\thesubfigure}{}
	\subfigure[(a) Naive approach]{
		\includegraphics[scale = 0.3]{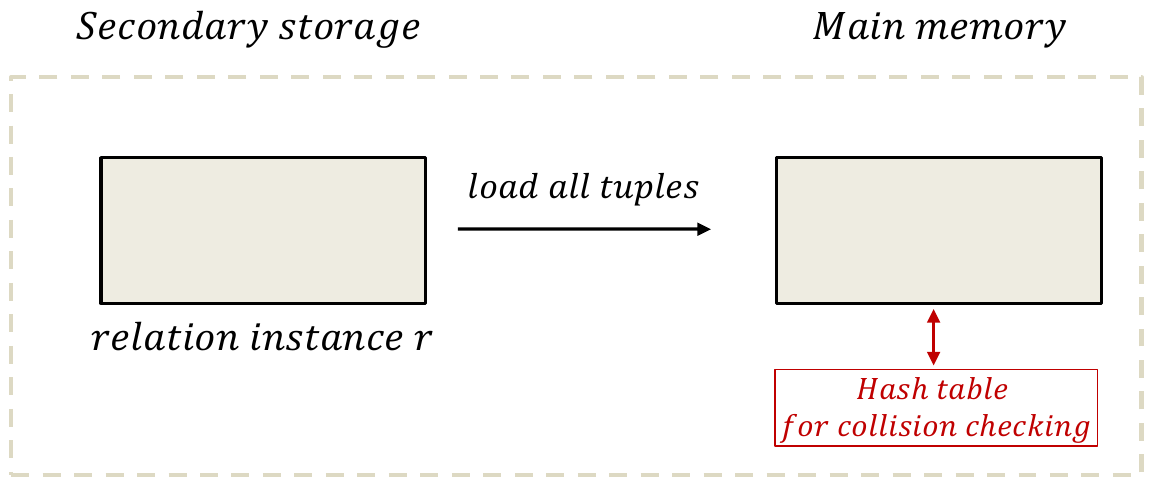}}
	\subfigure[(b) Partial-loading approach]{
		\includegraphics[scale = 0.3]{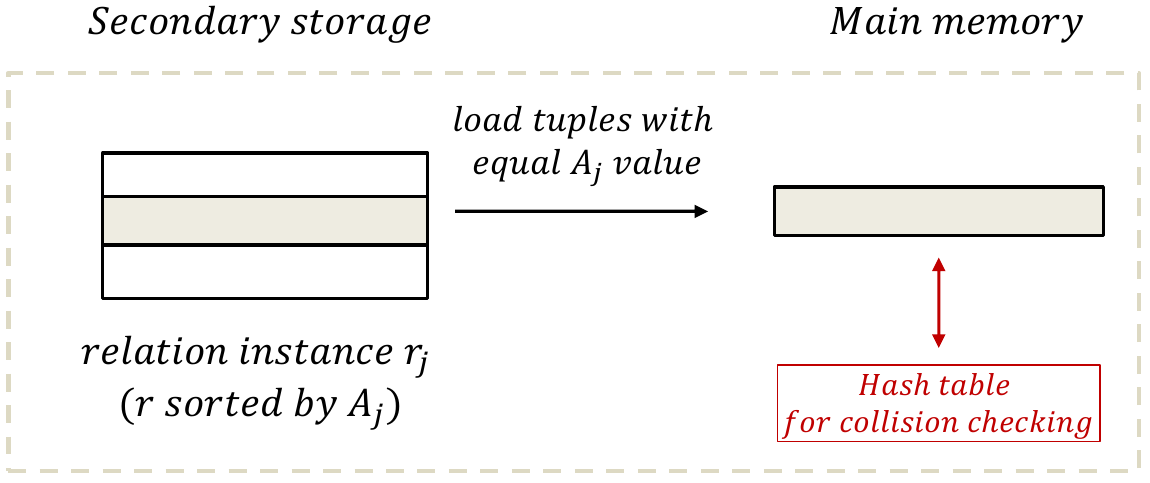}}
	\caption{The illustration of loading a subset of tuples for collision checking. (a) Naive approach of loading the entire relation instance \textit{r} into memory. (b) Partial-loading approach used in DUD, which uses the relation instance $r_j$ (i.e., $r$ sorted by attribute $A_j$) and loads only the tuples within the current run into memory. Either (a) or (b) is used to load tuples before hash-based collision checking, and DUD adopts (b).}
	\label{fig:partialloading}
\end{figure}

It is obvious that the validation candidates in $ST_c$ only contain attributes in $S_{NG}$, where $S_{NG}$ is the set of attributes of NG type. The validation process of phase 3 is executed in several iterations, each of which deals with an attribute in $S_{NG}$. The attributes in $S_{NG}$ is \textit{accessed in the ascending order of maximum run lengths}, i.e., the occurrence number of value that appears most often in an attribute, which are computed in phase 1. Given a candidate $X \in ST_c$, the sorted relation instance corresponding to any constituent attribute in \textit{X} can be used to validate \textit{X} gaining the benefit mentioned above but with varying degrees. The sorted relation instance with the smallest maximum run length is a better choice for validating \textit{X}, since it normally minimizes the maximum number of tuples with equal attribute value.

Let $A_j \in S_{NG}$ be the visited attribute in current iteration. The candidates in $ST_c$ including $A_j$ are kept in a set $C_j$ and removed from $ST_c$ (Line \ref{alg3:line9} in Algorithm \ref{alg:phase3}). If $C_j$ is empty, the next attribute in $S_{NG}$ is visited (Line \ref{alg3:line10} and Line \ref{alg3:line11}). Otherwise, DUD sorts relation instance \textit{r} in the ascending order of $A_j$ as $r_j$, if $r_j$ has not been materialized, and all candidates in $C_j$ are checked simultaneously by retrieving $r_j$ once only.

\textit{Hash-based batch validation by $r_j$}. $\forall X \in C_j$, the validation checks whether $r_j[X]$ contains duplicates. For candidate \textit{X}, a hash table $HT_X$ is built with the projection of tuples on \textit{X} as \textit{key} and a set of tuples as \textit{value} (Line \ref{alg3:line14} in Algorithm \ref{alg:phase3}). Given that $r_j$ is sorted in ascending order of $A_j$, and let $v_{j,1} \le v_{j,2} \le \ldots \le v_{j, |A_j|}$ be the distinct values occurred in $A_j$. Obviously the tuples with $A_j = v_{j,k}$ ($1 \le k \le |A_j|$) are arranged consecutively in $r_j$. DUD retrieves the tuples in $r_j$ sequentially. Let \textit{t} be the currently retrieved tuple. The following two cases are processed. 
\begin{itemize}
	\item[-] Case 1: \textit{t} has the same $A_j$ value as the previous tuple. If $HT_X$ does not contain an element whose key is $t[X]$, $(t[X], \{t\})$ is put to $HT_X$. Otherwise, \textit{t} is added to the set of tuples corresponding to key $t[X]$. The latter actually indicates that \textit{X} is a non-UCC since at least two tuples have the same projection on it. Given that $r_j$ is sorted with respect to $A_j$, the duplicates of projection on \textit{X} can only exist in the tuples with the same $A_j$ values.
	\item[-] Case 2: \textit{t} has the greater $A_j$ value than the previous tuple. This means that the following tuples cannot have duplicate projections on \textit{X} with the preceding tuples. At this point, any \textit{values} in $HT_X$ containing at least two elements are compared pairwise to generate difference sets, which is inserted into $\Delta_{DS}$ (Line \ref{alg3:line18} to Line \ref{alg3:line21} in Algorithm \ref{alg:phase3}). These tuples with the same projections on \textit{X} are \textit{conflicting tuples}. Then $HT_X$ is emptied, ready for the following operations.
\end{itemize}
Note that in validation process, DUD utilizes an array, whose size is maximum run length in $A_j$, to keep the tuples for current run $(A_j = v_{j, k})$, and the \textit{values} in hash table only maintain references to the tuples rather the tuples themselves. When a new attribute value occurs, the array is reused to keep the tuples of next run $(A_j = v_{j, k + 1})$. 

At the end of retrieval, if there does not exist any duplicate projection for candidate \textit{X}, it is reported to be a UCC and maintained in $ST_{ucc}$ (Line \ref{alg3:line22} in Algorithm \ref{alg:phase3}). If any candidates are validated to be non-UCCs, the difference sets in $\Delta_{DS}$ are used to update hypergraph $\mathcal{H}$ as new hyperedges.

\begin{figure}
	\centering
	\includegraphics[scale = 0.31]{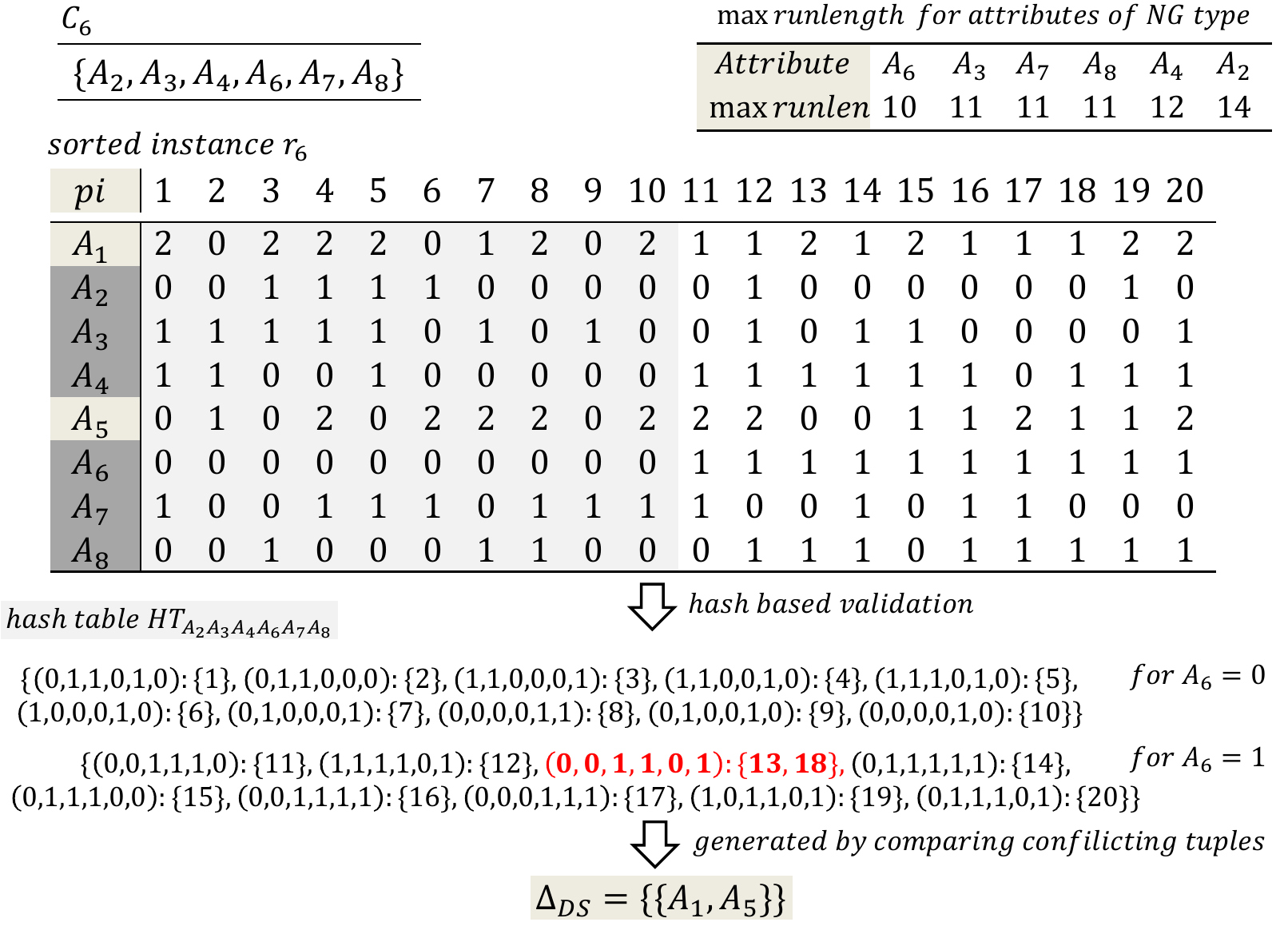}
	\caption{The illustration of hash-based validation for $C_6$.}
	\label{fig:visio_hashvalidation}
\end{figure}

\begin{example}
	Figure \ref{fig:visio_hashvalidation} illustrates the hash-based validation for UCC candidates in the running example. As depicted in Figure \ref{fig:visio_candidatepruning}, the only UCC candidate is $\{A_2$, $A_3$, $A_4$, $A_6$, $A_7$, $A_8\}$. Since $A_6$ has the smallest maximum run length among $S_{NG}$, DUD validates the candidate in $C_6$ by sorting \textit{r} as $r_6$ first. The hash table $HT_{A_2 A_3 A_4 A_6 A_7 A_8}$ is initialized to perform the collision checking. Due to the sortedness of $r_6$, the projections of at most 10 tuples are maintained in $HT_{A_2 A_3 A_4 A_6 A_7 A_8}$ instead of the total 20 tuples. By hash-based validation, the candidate in $C_6$ is found to be non-UCC and two conflicting tuples are maintained in the hash table. The corresponding difference set $\{ A_1, A_5 \}$ is generated and maintained in $\Delta_{DS}$.
\end{example}

The iteration on $S_{NG}$ continues until $ST_c$ is empty (Line \ref{alg3:line23} and Line \ref{alg3:line24} in Algorithm \ref{alg:phase3}). By now all candidates are validated. If $\Delta_{DS}$ is empty, it means all candidates are true UCCs, DUD ends and reports the found UCCs in $ST_{ucc}$ directly (Line \ref{alg1:line12} in Algorithm \ref{alg:dud}). Otherwise, DUD uses the difference sets in $\Delta_{DS}$ as newly added hyperedges in $\mathcal{H}$ and enters phase 2 again to enumerate minimal hitting sets $ST_c$ of the updated $\mathcal{H}$ (Line \ref{alg1:line:6} to Line \ref{alg1:line:8} in Algorithm \ref{alg:dud}). The following operation is executed similarly with the difference that the newly generated candidates already discovered to be UCCs in the previous operations should be removed, i.e., $ST_c = ST_c \setminus ST_{ucc}$. If $|ST_c| = 0$ or all candidates in $ST_c$ are validated to be true UCCs in phase 3, DUD ends and reports the UCCs in $ST_{ucc}$.

\begin{figure}
	\centering
	\includegraphics[scale = 0.31]{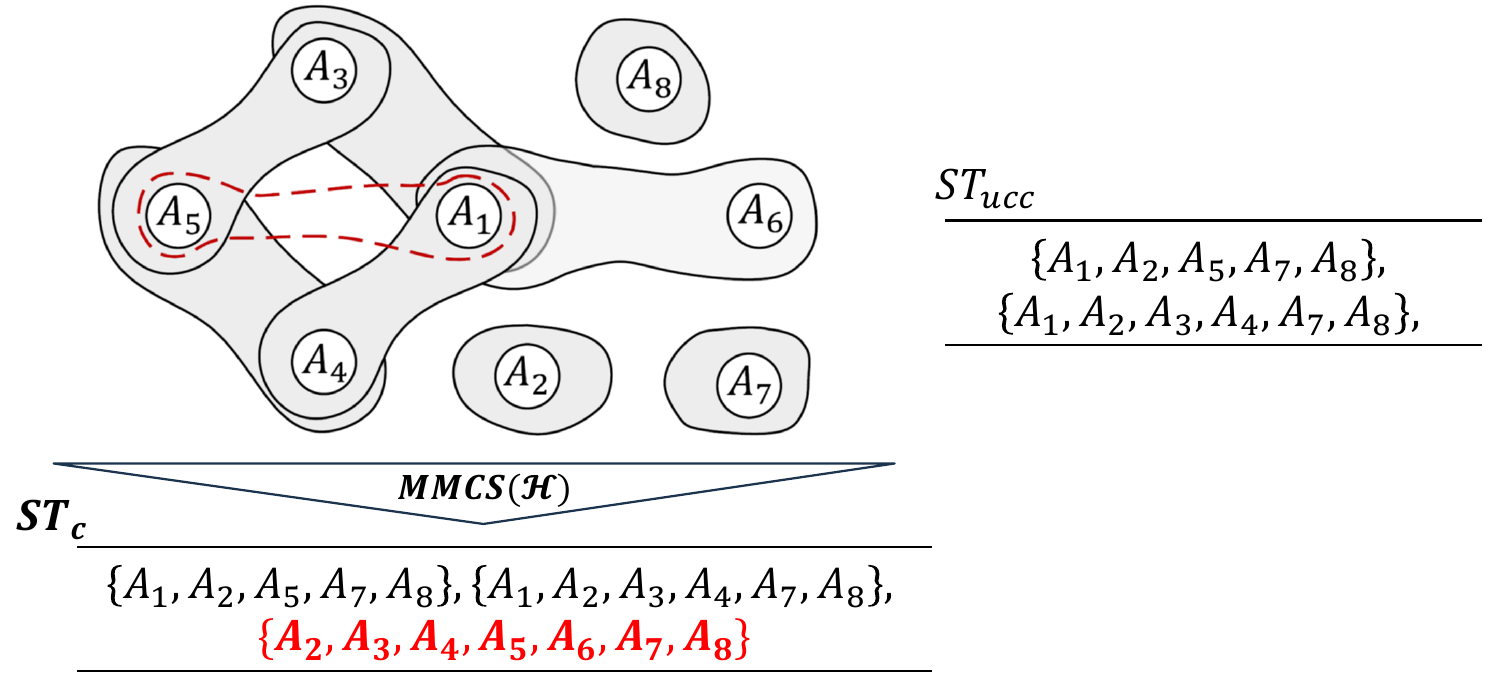}
	\caption{The illustration of second execution results of phase 2 on updated hypergraph $\mathcal{H}$ by adding a new hyperedge $\{A_1, A_5\}$.}
	\label{fig:visio_phase2_updated}
\end{figure}

\begin{example}
	Figure \ref{fig:visio_phase2_updated} illustrates second execution results of phase 2 on the updated hypergraph $\mathcal{H}$ by adding a new hyperedge $\{A_1, A_5\}$. Three candidates are generated on the updated $\mathcal{H}$, two of which are included in $ST_{ucc}$ and removed from the candidate set. The only candidate is $\{ A_2, A_3, A_4, A_5, A_6, A_7, A_8 \}$ is validated and reported as a true UCC. The execution of DUD is over.
\end{example}

\textit{Implementation of the I/O procedure}. In phase 1, the relation is first decomposed into $m$ column files, with one read of the relation and one write of the column files. Each column file is sorted to compute $cmp_i$ and determine whether attribute $A_i$ is of FG or NG type, sorting single-column files rather than the full relation. For each FG attribute, the relation is sorted by that attribute and scanned to generate full useful difference sets. Random sampling accesses a subset of tuples from the relation to supplement the initial hypergraph. If a file to be sorted exceeds the sorting buffer size, the sort requires TPMMS, with additional sequential reads and writes of temporary files. Otherwise the file is loaded into memory, sorted, and written back to disk without intermediate temporary files. In both cases, the sorted result is materialized as a file on disk rather than kept in memory, so that it can be read by subsequent steps.

In phase 3, candidates are grouped by NG attribute into batches $C_j$, and one scan of the relation sorted by $A_j$ validates all candidates in $C_j$ together, rather than scanning the relation separately for each candidate. Re-entering phase 2 after updating the hypergraph does not access the relation at all, since MMCS operates only on $\mathcal{H}$. If phase 3 is invoked again, NG attributes are visited in the same ascending order of maximum run length. Only those with a non-empty $C_j$ trigger a scan of the sorted relation. If the sorted relation was already materialized in an earlier iteration, it is reused rather than re-sorted. In our experiments, phase 3 is invoked at most twice, and the proportion of candidates reported as true UCCs without validation often exceeds 90\%, so few NG attributes ever require sorting.

\textit{Relation to direct validation in HyFD}. The batch validation strategy of DUD is related in spirit to the direct validation procedure of HyFD \cite{DBLP:conf/sigmod/PapenbrockN16}, which also validates multiple dependency candidates together by using a PLI of an LHS attribute with many small clusters and mapping LHS clusters to RHS clusters. This is similar in purpose to the use of an NG attribute with a small maximum run length in DUD to reduce the number of tuples kept in memory. However, the validation of HyFD relies on PLIs and \textit{pliRecords} built during preprocessing and held in memory, whereas DUD does not maintain such global PLI-based structures: it scans the relation sorted by the chosen NG attribute directly from disk and keeps only the tuples of the current run in memory.

\subsection{The correctness of DUD}

This part provides the proof for the correctness of DUD algorithm.

Let $\mathcal{H}_{A}(V, E_{A})$ be the hypergraph whose vertices are attributes of relation schema and hyperedges are the difference sets generated by all pairwise tuple comparisons. Obviously, the minimal hitting sets in $\mathcal{H}_{A}$ are the minimal UCCs surely. Any hypergraph $\mathcal{H}(V, E)$ used in DUD actually is a spanning subgraph of $\mathcal{H}_{A}$ \cite{DBLP:books/daglib/0030488}, where $E \subseteq E_{A}$.

\begin{theorem} \label{theorem:full2partial}
	Given $\mathcal{H}_A(V, E_A)$ and its spanning subgraph $\mathcal{H}(V, E)$, $\forall h_A \in Tr(\mathcal{H}_A)$, $\exists h \in Tr(\mathcal{H})$, $h_A \supseteq h$.
\end{theorem}
\begin{proof}
	According to definition of minimal hitting sets, $\forall h_A \in Tr(\mathcal{H}_A)$, it satisfies: $\forall e_A \in E_A$, $h_A \cap e_A \neq \emptyset$. We can treat the requirements above to be satisfied as the constraints for $h_A$. Obviously, if $\exists e_A \in E_A$, $\exists p \in \mathcal{P}(V)$, and $e_A \cap p = \emptyset$, an element in $e_A$ has to be added into \textit{p} to make it be minimal hitting set. In other words, any minimal hitting set is expanded from empty set to have at least one common element with every hypergraph. Since $E \subseteq E_A$, any minimal hitting set \textit{h} enumerated in $\mathcal{H}$ can be seen in intermediate stage of the enumeration process in $\mathcal{H}_A$. To sum up, $\forall h_A \in Tr(\mathcal{H}_A)$, there exists an element in $Tr(\mathcal{H})$, which is generated in the middle of expansion process, and $h_A \supseteq h$.
\end{proof}

\begin{theorem} \label{theorem:partial2full}
	Given $\mathcal{H}_A(V, E_A)$ and its spanning subgraph $\mathcal{H}(V, E)$, $\forall h \in Tr(\mathcal{H})$, $\exists h_A \in Tr(\mathcal{H}_A)$, $h \subseteq h_A$.
\end{theorem}
\begin{proof}
	With the similar argument of Theorem \ref{theorem:full2partial}, $\forall h \in \mathcal{H}$, if hyperedges in $E_A \setminus E$ are added to $\mathcal{H}$, $h$ can be expanded to be a minimal hitting set $h_A \in Tr(\mathcal{H}_A)$, i.e., $h \subseteq h_A$.
\end{proof}

\begin{figure*}
	\centering
	\includegraphics[scale = 0.45]{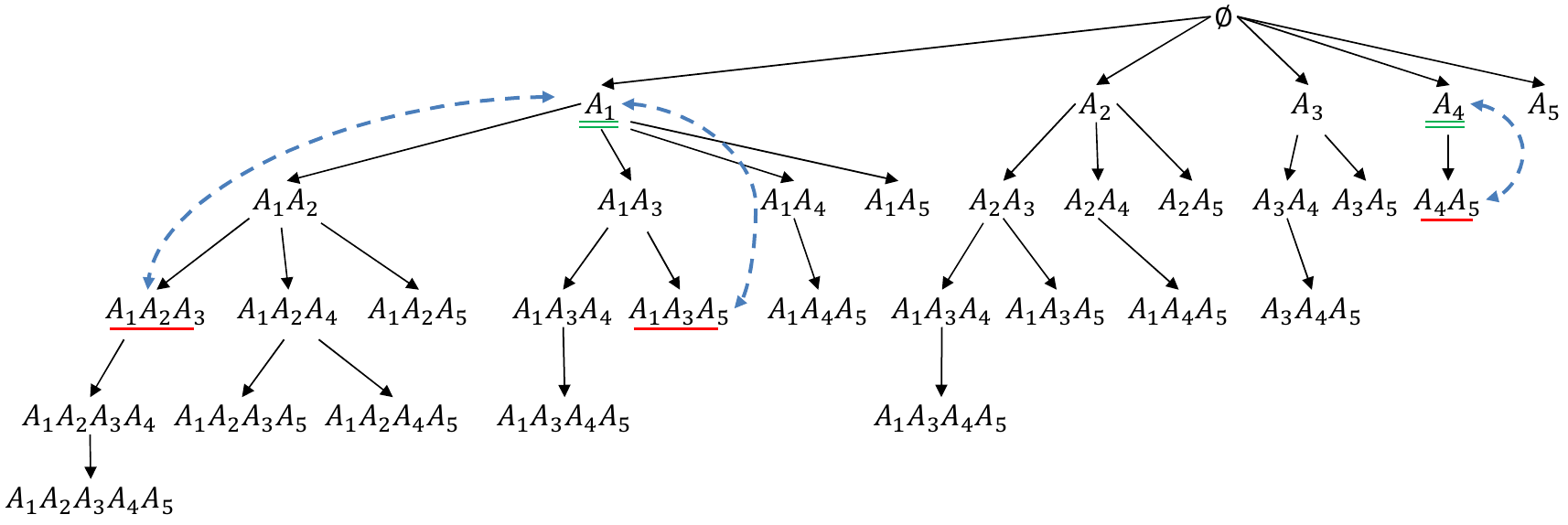}
	\caption{The illustration of correspondence between minimal hitting sets in hypergraph used by DUD and complete hypergraph. The attribute sets underlined with red line are final minimal hitting sets and those underlined with green double line are minimal hitting sets in hypergraph used by DUD.}
	\label{fig:visio_full2partial}
\end{figure*}

Combining Theorem \ref{theorem:full2partial} and Theorem \ref{theorem:partial2full}, every minimal hitting set in $\mathcal{H}_A$ corresponds to a minimal hitting set in $\mathcal{H}$, and every minimal hitting set in $\mathcal{H}$ has a corresponding final minimal hitting set in $\mathcal{H}_A$, which is illustrated in Figure \ref{fig:visio_full2partial}. Theorem \ref{theorem:correctness} proves the correctness of DUD.

\begin{theorem} \label{theorem:correctness}
	DUD discovers the true UCCs correctly.
\end{theorem}
\begin{proof}	
	Given hypergraph $\mathcal{H}$ used by DUD, $\forall h \in Tr(\mathcal{H})$, as proved in Theorem \ref{theorem:partial2full}, $\exists h_A \in Tr(\mathcal{H}_A)$, $h \subseteq h_A$. If $h = h_A$, it will be reported as true UCCs either by Theorem \ref{theorem:candidatepruning} or hash-based validation. If $h \subset h_A$, since $h_A$ is a minimal hitting set of $\mathcal{H}_A$, i.e., minimal UCC in relation instance \textit{r}, \textit{h} is a non-UCC and the validation process for \textit{h} can find conflicting tuples, which have equal projections on \textit{h}. At least one hyperedge $e_{new}$ will be added to $\mathcal{H}$, which excludes attributes in \textit{h}. In the next execution of phase 2, DUD enumerates the minimal hitting sets on the updated $\mathcal{H}$. For \textit{h} in the current enumeration, \textit{h} should have at least one element with $e_{new}$ and at least one more attribute in $h_A \setminus h$ can be expanded to \textit{h}. The expansion can be continued until $h = h_A$. This guarantees that discovered UCCs in DUD are true UCCs. As proved in Theorem \ref{theorem:full2partial}, $\forall h_A \in Tr(\mathcal{H}_A)$, $\exists h \in Tr(\mathcal{H})$, $h_A \supseteq h$, this actually guarantees that every true UCC will be discovered by DUD. To sum up, DUD discovers the true UCCs correctly.
\end{proof}

\subsection{Cost analysis}

This part analyzes the time complexity of DUD. Three phases are discussed individually and then combined to get the overall complexity.

In phase 1, the execution of DUD consists of several operations sequentially: decomposing relation instance, sorting column files and computing $\gamma$, sorting relation instances and generating difference sets with respect to attributes of FG type, supplement of difference sets by random sampling. Given that $\gamma = 3 \times n \times \log_2n$, DUD generates difference sets within $O(m \times n \times \log_2 n)$ time. And the sorting operation also has a $O(m \times n \times \log_2 n)$ time complexity. The other operations in phase 1 are of linear complexity. Therefore, the time complexity of phase 1 is $O(m \times n \times \log_2 n)$.

The execution of phase 2 is a breadth-first exploration in set enumeration tree $\mathcal{T}$ of \textit{m} attributes with the specified pruning operation. It is proved that given hypergraph $\mathcal{H}(V, E)$, MMCS enumerates minimal hitting sets within $O(\left\Vert E \right\Vert)$ for current visited node in $\mathcal{T}$, where $\left\Vert E \right\Vert$ is the sum of size of hyperedges in \textit{E} \cite{DBLP:journals/dam/MurakamiU14}. Let $a$ be maximum cardinality of any minimal hitting sets, the number of visited nodes in $\mathcal{T}$ is at most $\sum_{i = 0}^{a} \binom{m}{i}$. The time complexity of phase 2 is $O(\left\Vert E \right\Vert \times \sum_{i = 0}^{a} \binom{m}{i})$.

Phase 3 involves sorting operations with respect to attributes of NG type, scans the sorted relation instances sequentially and checks collisions, where the new difference sets are generated by comparing conflicting tuples pairwise. Let $c$ be the maximum distinct value number in attributes of NG type, and $d$ be the maximum number of duplicate projections on candidates in any run. The time complexity of phase 3 is $O(m \times n \times \log_2 n + m \times c \times d^2 \times (1 - P_{pru}) \times \sum_{i = 0}^{a} \binom{m}{i})$. Here, $P_{pru}$ is the proportion of candidates generated in phase 2 which can be reported as true UCCs without validation.

Let $b$ be the number of invocation of phase 3. Overall, taking into account three phases, DUD has a $O(b \times m \times n \times \log_2 n + b \times \left\Vert E \right\Vert \times \sum_{i = 0}^{a} \binom{m}{i} + b \times m \times c \times d^2 \times (1 - P_{pru}) \times \sum_{i = 0}^{a} \binom{m}{i})$ time complexity. It is shown in \cite{DBLP:journals/jcss/BlasiusFLMS22} that the enumeration of minimal hitting sets normally runs fast on hypergraphs stemming from real-life data. This is also verified in the experiments that on all used real-life data sets, maximum size of UCCs is 12, and maximum number of minimal difference sets are 2649. The execution time of phase 2 normally is dominated by other two phases on real-life data sets. The invocation number of phase 3 is rather small, at most 2 in the experiments. The number of \textit{d} usually is very small, since the validated candidates are close to actual UCCs, and they only have a small number of duplicate projection values. Besides, the value of $P_{pru}$ often is greater than 90$\%$ and even up to 100$\%$. On real-life data sets, the first factor $O(b \times m \times n \times \log_2 n)$ is the dominant part in overall time complexity of DUD.

\section{Experimental evaluation} \label{sec:experiments}

This section evaluates the performance of DUD, which is implemented in Java with jdk-21\_windows-x64. The experiments are run on DELL OptiPlex 7010MT Workstation (Intel(R) Core(TM) i9-13900 CPU @ 2.00GHz (24 cores) + 32G memory + 64bit windows 11 + 4T HDD). The performance of DUD is evaluated against HPIValid, which is the state-of-the-art UCC discovery algorithm and is reported to be orders of magnitude faster than related work with a much smaller memory footprint \cite{DBLP:journals/pvldb/BirnickBFNPS20}. We reproduce HPIValid in Java according to the implementation in \cite{DBLP:journals/pvldb/BirnickBFNPS20}. The hash tables and sets involved in the execution utilize HashMap and HashSet in the Java.util package. The sampling exponent is set to 0.3, which is proved to be a good choice for HPIValid \cite{DBLP:journals/pvldb/BirnickBFNPS20}.

In the experiments, we evaluate the performance of DUD on synthetic data sets and real-life data sets. The synthetic data sets (the \textit{lineitem} table) are generated by TPC-H Benchmark\footnote{https://www.tpc.org/tpch/}, while the real-life data sets are downloaded from UCI Machine Learning Repository\footnote{http://archive.ics.uci.edu/} and Kaggle\footnote{https://www.kaggle.com/} to evaluate the performance of DUD in the practical applications. In the experiments, 24GB memory (32GB in total) is allocated for the algorithm execution. The \textit{-Xmx} parameter is used to specify the maximum memory an algorithm can use.

Every data set is preprocessed in the same way before being given to DUD or HPIValid. For each attribute, values are read as strings, and distinct values are collected and mapped to consecutive integer IDs, which then replace the original values. Since UCC discovery only requires checking whether two values on an attribute are equal, this loses no information needed by either algorithm. The resulting table is stored on disk with each attribute value encoded as a fixed-length 8-byte integer, the format read by both DUD and HPIValid throughout the experiments.

The parameters of the lineitem table are described in Table \ref{table:parameter}, whose attribute number is 16 and tuple numbers are a multiple of scale factor (SF), i.e., SF $\times$ 6,000,000 approximately.

Given SF = 10, Table \ref{table:partialgenerationforsf10} provides the details of the lineitem table, including cardinalities, full useful comparison numbers and types of attributes. The cost of computing full useful difference sets with respect to some attributes is prohibitively expensive, as listed in Table \ref{table:partialgenerationforsf10}. Here, we only provide the comparison information about the lineitem table with SF = 10, and the lineitem tables of other SF values have the similar situation. In the experiments, given $\gamma = 3 \times n \times \log_2 n$, the first, second, sixth and sixteenth attributes in the lineitem table are of FG type, and other attributes are of NG type. It is worth noting that the number of distinct difference sets generated for each attribute is significantly less than the number of full useful pairwise comparisons. For example, for the sixth attribute in Table \ref{table:partialgenerationforsf10}, the number of full useful pairwise comparisons is 1,853,382,718, while only 1965 distinct difference sets are generated. That is to say, about $10^6$ pairwise comparisons can generate a new difference set for the sixth attribute. Note that number of distinct difference sets with respect to the specified attribute in Table \ref{table:partialgenerationforsf10} is not the number of minimal difference sets, but just the number of difference sets after deduplication. Given SF = 10, the number of minimal difference sets generated with respect to four attributes of FG type is 198, which is much smaller than that of distinct difference sets.

\textit{NULL semantics}. Some real-life data sets have incomplete information, and a null value is used to indicate a missing or unknown value. This paper considers the commonly used null semantics (\textit{NULL-EQ}), which treats null values as identical and not equal to any non-null value \cite{DBLP:journals/pvldb/Berti-EquilleHN18,DBLP:journals/pvldb/BirnickBFNPS20}. We adopt NULL-EQ to keep the UCC definition consistent with HPIValid and other UCC discovery studies, enabling a direct comparison under the same semantics. Under SQL-style uniqueness semantics, where NULL values are treated as distinct, tuples with NULL values on an attribute would no longer be grouped into the same equivalence class, generally reducing $cmp_i$ and the size of validation batches in phase 3. We have not implemented or evaluated DUD under this alternative semantics.

\begin{table}[h]
	\centering
	\caption{Parameter information of lineitem in TPC-H, where attribute number is 16, $\gamma = 3 \times n \times \log_2 n$, SF means scale factor and UCCNum is the number of discovered UCCs.}
	\small
	\begin{tabular}{c|cccc}
		\hline
		SF & TupleNum & $\gamma$ value & UCCNum\\
		\hline
		1 & 6,001,215 & 405,384,891 & 390 \\ 
		2 & 11,997,996 & 846,445,064 & 404\\
		5 & 29,999,795 & 2,235,445,161 & 364\\
		10 & 59,986,052 & 4,649,781,104 & 342\\
		15 & 89,987,373 & 7,133,268,257 & 318\\
		20 & 119,994,608 & 9,661,387,823 & 358\\
		\hline
	\end{tabular}
	\label{table:parameter}
\end{table}

\begin{table}[h]
	\centering
	\caption{Cardinalities and types of attributes in lineitem, where SF = 10 and $\gamma = 4,649,781,104$, DSN is the number of distinct difference sets generated with respect to the specified attribute.}
	\small
	\begin{tabular}{c|cccc}
		\hline
		Attr & Cardinality & TotalComparison & Type & DSN \\
		\hline
		1 & 15,000,000 & 119,948,344 & \textbf{FG} & 757 \\ 
		2 & 2,000,000 & 898,711,177 & \textbf{FG} & 1719\\
		3 & 100,000 & 17,991,006,982 & NG & -\\
		4 & 7 & 321,330,030,852,817 & NG & -\\
		5 & 50 & 35,983,264,280,805 & NG & -\\
		6 & 1,351,462 & 1,853,382,718 & \textbf{FG} & 1965\\
		7 & 11 & 163,560,278,023,276 & NG & -\\ 
		8 & 9 & 199,907,039,167,641 & NG & -\\
		9 & 3 & 680,505,080,028,315 & NG & -\\
		10 & 2 & 899,581,606,017,474 & NG & -\\
		11 & 2,526 & 735,235,731,620 & NG & -\\
		12 & 2,466 & 741,466,394,917 & NG & -\\
		13 & 2,555 & 734,875,394,017 & NG & -\\ 
		14 & 4 & 449,790,785,523,119 & NG & -\\
		15 & 7 & 257,023,299,207,497 & NG & -\\
		16 & 34,378,943 & 2,270,982,389 & \textbf{FG} & 1555\\
		\hline
	\end{tabular}
	\label{table:partialgenerationforsf10}
\end{table}

\begin{figure}
	\centering
	\includegraphics[scale = 0.8]{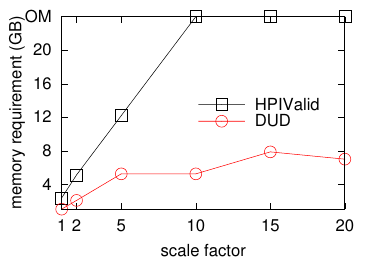}
	\caption{Memory space requirement.}
	\label{fig:memoryspace}
\end{figure}

\subsection{Memory Space Requirement} \label{sec:exp:memoryspace}

This part illustrates the memory space required for HPIValid. Given the lineitem tables of SF = 1, 2, 5, 10, 15 and 20, Figure \ref{fig:memoryspace} illustrates the memory consumption under the space limit 24GB. It is shown that HPIValid requires about 2.40GB, 5.14GB and 12.21GB for lineitem tables with SF = 1, 2 and 5 respectively, and runs out of memory for SF $\ge$ 10. ``OM'' in Figure \ref{fig:memoryspace} indicates out of memory. In our implemented HPIValid, the record-ID-to-cluster-ID inverse mapping required for validation is represented by the same primitive \texttt{int[][]} array used to store the cluster ID of each tuple on each attribute, rather than a separate hash table. The array is built during pre-processing and used both for difference-set computation and for validation.

The reported memory measurements are obtained by first triggering garbage collection and then computing the difference between total and free JVM heap memory. Without this step, objects that are no longer used but have not yet been collected would still appear as used memory, making the reported number higher than the actual memory in use. To avoid triggering garbage collection too often, it is only performed when the run is configured to report memory consumption, and only at points where heap usage exceeds the highest value recorded so far.

For comparison, the memory space requirement for DUD also is depicted in Figure \ref{fig:memoryspace}. For DUD, it only needs to sort the column files and relation instance, keep the difference sets, build the hypergraph, maintain tuples in each run and the hash tables which store projections on the UCC candidates for the tuples with equal attribute values to check duplicates. These operations are executed sequentially and independently. Hence, DUD has a lower memory overhead, less than 8GB memory at SF $\le$ 20. In the experiments, we limit the maximum size of each sorted partition to be one quarter of the allocated memory. Since the sorting operation occupies the largest memory space in the execution, the memory required for sorting remains bounded on data sets of larger scale. The remaining variation in memory consumption of DUD across different SF values comes from data-dependent factors including the size of difference sets, the candidate set, and the hash tables used during validation. We discuss this issue further in Section \ref{sec:exp:discussion}.

\begin{figure}
	\centering
	\includegraphics[scale = 0.8]{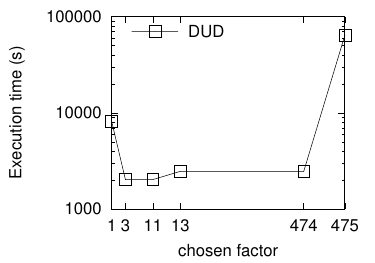}
	\caption{Different factor selections for $\gamma$ value, the experiments are executed on the lineitem table of SF = 10.}
	\label{fig:factorChoice}
\end{figure}

\begin{table}[h]
	\centering
	\caption{Execution time in seconds of three phases of DUD with different factors, which is the multiples of $n \times \log_2 n$}
	\small
	\begin{tabular}{c|cccc}
		\hline
		Factor & Phase1 & Phase 2 & Phase 3 & $S_{FG}$ \\
		
		\hline
		1 & 610.55 & 0.09 & 7571.42 & 1,2\\
		3 & 1217.01 & 0.11 & 811.69 & 1,2,6,16\\
		13 & 2130.65 & 0.11 & 342.34 & 1,2,3,6,16\\
		475 & 64123.34 & 0.04 & 0.004 & 1,2,3,6,11,13,16\\
		\hline
	\end{tabular}
	\label{table:factorChoiceDetail}
\end{table}

\subsection{Factor selection of $\gamma$ value} \label{sec:exp:factorselection}

In this paper, DUD uses $3 \times n \times \log_2n$ as $\gamma$ value. The choice of $n \times \log_2 n$ is relatively straightforward since its growing rate is substantially lower than $O(n^2)$, which is the upper-bound of pairwise comparison number. The choice of factor 3 is explained in this part through experimental verification on the lineitem table of SF = 10. The selected factors include 1, 3, 11, 13, 474 and 475, which comprehensively reflects the different cases for DUD execution. As illustrated in Figure \ref{fig:factorChoice}, the execution time of DUD first decreases, next remains unchanged, then increases slightly and remains unchanged again, last shows a rapid growth. The time decomposition of three phases for different factors is listed in Table \ref{table:factorChoiceDetail}. 

As we expect, the execution time of phase 2 is negligible in the experiments. When the factor is lower, fewer attributes are of FG type. In the case, phase 1 terminates quickly, but phase 3 needs a longer time to validate the candidates because the pruning effect of Theorem \ref{theorem:candidatepruning} declines. On the contrary, when the factor is higher, more attributes are of FG type. Of course, phase 1 may require a much longer time to terminate, but phase 3 may end quickly. 

It is also shown that the execution time of DUD remains unchanged across a range of selected factors. This can be explained that, the required full useful pairwise comparison numbers have relatively large gaps for different attributes. If the $\gamma$ values for the selected factors lie in the same gap, they do not affect the execution of DUD. For example, when the factor is chosen in $[3, 11]$, there are four attributes of FG type.

Taken together, when factor for $\gamma$ value is set to 3, the overall performance of DUD does the best in Figure \ref{fig:factorChoice}, and $\gamma$ is set to $3 \times n \times \log_2 n$ in the experiments.

\begin{figure}
	\centering
	\renewcommand{\thesubfigure}{}
	\subfigure[(a) Execution time]{
		\includegraphics[scale = 0.63]{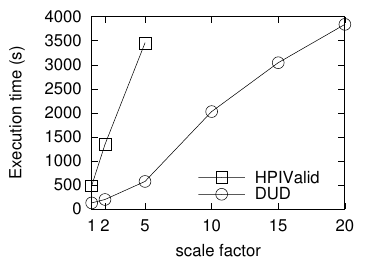}}
	\subfigure[(b) Time of three phases]{
		\includegraphics[scale = 0.63]{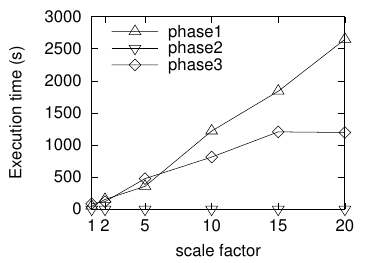}}
	\subfigure[(c) Retrieved I/O cost]{
		\includegraphics[scale = 0.63]{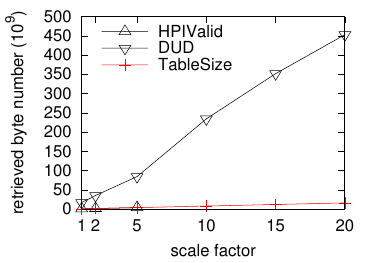}}
	\subfigure[(d) Pairwise comparison]{
		\includegraphics[scale = 0.63]{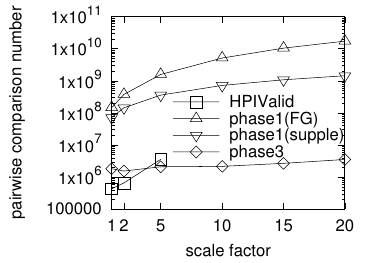}}	
	\subfigure[(e) UCC candidate number]{
		\includegraphics[scale = 0.63]{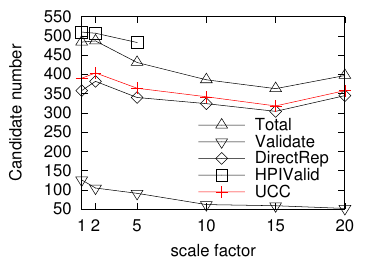}}	
	\subfigure[(f) Maintained tuple number]{
		\includegraphics[scale = 0.63]{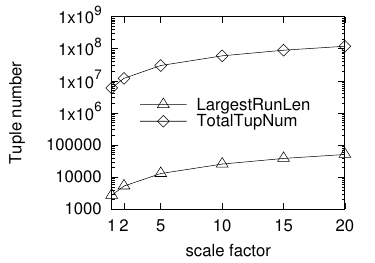}}	
	\caption{The illustration of varying tuple number.}
	\label{fig:tupnum}
\end{figure}

\begin{table}[h]
	\centering
	\caption{Time decomposition of phase 1 in experiment 1, where ``Decom" is the operation to decompose row-table into column files, ``DetType" is the operation to sort column files, compute $\gamma$ value and determine attribute type, ``Sort" is the operation to sort relation instance with respect to attributes of FG type, ``GenDS" is the operation to generate full useful difference sets with respect to attributes of FG type, ``Supple" is the operation to sample relation instance to supplement the difference sets.}
	\small
	\begin{tabular}{c|ccccc}
		\hline
		SF & Decom & DetType & Sort & GenDS & Supple\\
		\hline
		1 & 8.527 & 15.817 & 23.272 & 11.525 & 3.433 \\ 
		2 & 13.994 & 38.596 & 62.855 & 39.699 & 7.399 \\
		5 & 81.433 & 51.39 & 133.468 & 74.347 & 17.12\\
		10 & 232.753 & 102.72 & 545.702 & 266.968 & 33.269\\
		15 & 304.975 & 159.96 & 888.105 & 526.109 & 49.115\\
		20 & 374.24 & 219.773 & 1189.126 & 815.548 & 65.391\\
		\hline
	\end{tabular}
	\label{table:tupnumphase1decomposition}
\end{table}

\subsection{Experiment 1: effect of varying tuple numbers}

Experiment 1 evaluates the performance of DUD on the lineitem tables of SF = 1, 2, 5, 10, 15 and 20, respectively. Since the memory space required by HPIValid exceeds the allocated memory capacity for SF $\ge$ 10, experiment 1 reports the experimental results of HPIValid with SF = 1, 2 and 5. 

As shown in Figure \ref{fig:tupnum}(a), with a greater value of SF, execution time of DUD increases gradually, but with a much lower growing trend compared with HPIValid. The decomposition of execution times of three phases is reported in Figure \ref{fig:tupnum}(b), and the further decomposition of phase 1 is listed in Table \ref{table:tupnumphase1decomposition}. By the reported results, DUD runs faster than HPIValid on the lineitem table of SF $\le$ 5, with speedups ranging from 3.7$\times$ at SF = 1 to 6.0$\times$ at SF = 5. The significant speedup is attributed to several factors. Firstly, through sorting primitive, DUD determines the attributes of FG type and builds the useful difference sets with a relatively low cost in phase 1. Furthermore, candidate pruning and hash-based batch validation in DUD go quite a long way. Most of candidates can be reported as true UCCs directly without validation, and several candidates can be verified in one iteration. As shown in Figure \ref{fig:tupnum}(e), about 90$\%$ candidates generated by phase 2 are removed from validation since they are proved to be UCC by Theorem \ref{theorem:candidatepruning}. This reduces the validation cost significantly. Just as we anticipate, the execution cost of phase 2, i.e., minimal hitting set enumeration, is very low, less than 0.2s in experiment 1. This also verifies rationality of repeatedly calling minimal hitting set enumeration to solve the problem incurred by possible partial generation of difference sets. 

As shown in Figure \ref{fig:tupnum}(c), DUD retrieves one order of magnitude times more data than the original table. It is interesting to find that \textit{high cost of memory computation often exceeds the cost of multiple sequential scans on disk}. The I/O procedure of DUD is detailed in Section \ref{sec:dud}. In this paper, sorting is treated as an individual operation, which sorts the input file to output sorted file, and the following operation retrieves the outputted sorted file directly. Although this increases the involved I/O cost, it simplifies the algorithm execution.

Due to the enough number of the computed difference sets, the generated UCC candidates are true UCCs or are close to true UCCs, and the latter only have a small number of duplicate projection values. The required numbers of pairwise comparisons are illustrated in Figure \ref{fig:tupnum}(d). The pairwise comparison number in phase 3 is rather low compared with that in phase 1, this also indicates the number of duplicate projection values in phase 3 is very low. 

As illustrated in Figure \ref{fig:tupnum}(f), the maximum number of tuples maintained in validation at a time is three orders of magnitude fewer than total tuple number. Even at SF = 20, the lineitem table has 119,994,608 tuples and only at most 50,554 tuples are maintained in memory for validation at a time.

For SF = 1, 2 and 5, the pairwise comparison numbers of HPIValid are reported in Figure \ref{fig:tupnum}(d). Obviously, the total comparison number of DUD is much larger than that of HPIValid. The overwhelming majority of pairwise comparisons of DUD are executed in phase 1, i.e., generating all useful difference sets with respect to attributes of FG type and supplement of difference by sampling. When performed on data set of small and medium scale, HPIValid has a lower I/O cost than DUD  and a less pairwise computation number. However, it still runs several times slower than DUD at SF = 1, 2 and 5. This can be explained for a variety of reasons. Firstly, the less pairwise comparison number does not mean a lower computation cost, but means a small number of duplicate projections on the candidates. Secondly, HPIValid validates each candidate via stripped partitions, whose cost is nontrivial, while DUD adopts a batch validation to check candidates in single validation process. Thirdly, HPIValid has to validate much more candidates than DUD, as shown in Figure \ref{fig:tupnum}(e), since it lacks a candidate pruning operation as in DUD. 

\begin{figure}
	\centering
	\renewcommand{\thesubfigure}{}
	\subfigure[(a) Execution time]{
		\includegraphics[scale = 0.63]{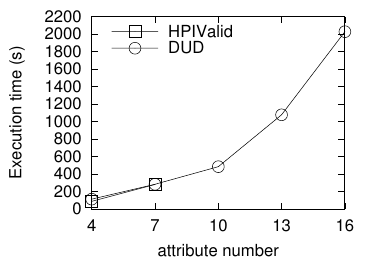}}
	\subfigure[(b) Time of three phases]{
		\includegraphics[scale = 0.63]{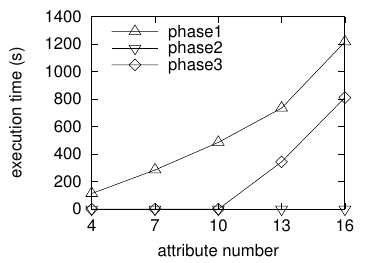}}
	\subfigure[(c) Retrieved I/O cost]{
		\includegraphics[scale = 0.63]{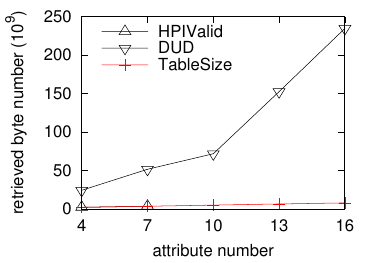}}
	\subfigure[(d) Pairwise comparison]{
		\includegraphics[scale = 0.63]{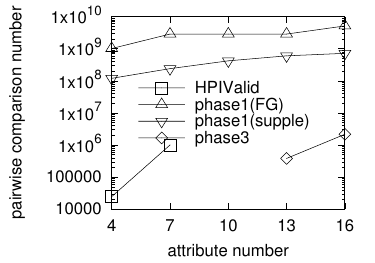}}	
	\subfigure[(e) UCC candidate number]{
		\includegraphics[scale = 0.63]{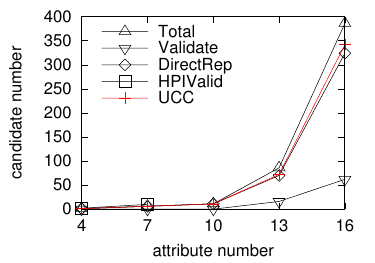}}	
	\subfigure[(f) Maintained tuple number]{
		\includegraphics[scale = 0.63]{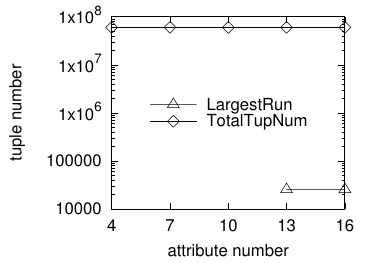}}	
	\caption{The illustration of varying attribute number.}
	\label{fig:attnum}
\end{figure}

\subsection{Experiment 2: effect of varying attribute numbers}

Experiment 2 evaluates the performance of DUD on the lineitem tables of SF = 10 with the first 4, 7, 10, 13 and 16 attributes, respectively. Since HPIValid runs out of memory for attribute numbers of 10, 13 and 16, the experimental results of HPIValid are only reported for attribute numbers of 4 and 7.

Figure \ref{fig:attnum}(a) illustrates the execution time of DUD in experiment 2. It is interesting to note that when attribute number is no greater than 10, DUD does not invoke hash-based validation at all. The reason is that, in such cases, all UCC candidates generated in phase 2 contain at least one attribute of FG type, and are proved to be true UCC according to Theorem \ref{theorem:candidatepruning}. Thus, no further validation is required for these candidates. This is a great advantage of DUD. For the data sets whose UCC candidates contain attributes of FG type, the UCCs can be discovered quickly, since phase 3 is not invoked at all and the hash-based validation does not need to be executed. This also is verified in Figure \ref{fig:attnum}(b), Figure \ref{fig:attnum}(c) and Figure \ref{fig:attnum}(d). 

As shown in Figure \ref{fig:attnum}(a), HPIValid completes for attribute numbers of 4 and 7, but runs out of memory for larger attribute numbers. When $m = 4$, HPIValid runs faster than DUD (88.70s vs. 113.86s). The PLI-based structures fit within the allocated memory, so HPIValid can validate candidates directly from memory without the disk-based sorting cost that DUD incurs. When $m = 7$, the two algorithms take similar time (281.67s vs. 284.79s). For $m \ge 10$, HPIValid runs out of memory while DUD still completes.

When attribute number increases from 10 to 13, there exist some UCC candidates only consisting of attributes of NG type, the execution time increases abruptly since the validation operation has to be performed. As shown in Figure \ref{fig:attnum}(e), with more attributes involved, the numbers of total candidates and true UCCs increase quickly. Correspondingly, more candidates are reported as true UCCs directly without validation, and the number of candidates to be validated shows a smaller increase trend.

The maximum number of tuples that are maintained in validation at a time is illustrated in Figure \ref{fig:attnum}(f). Compared with total tuples, DUD only needs to keep a much small number of tuples in memory. The missing points for number of maintained tuples in Figure \ref{fig:attnum}(f) are due to the fact that DUD does not invoke validation when the attribute number is no more than 10.

\begin{table*}[thp]
	\centering
	\caption{Experimental results on real-life data sets (runtime in seconds for HPIValid and DUD). The boldfaced results for DUD mean that the performance of DUD on these data sets is better than HPIValid. Breakdown shows the proportions of three phases of DUD, where yellow, green and red regions correspond to phase 1, phase 2 and phase 3, respectively. $|S_{FG}|$ is the number of attributes of FG type. OM is the abbreviation of OutOfMemoryError. Some original data sets (including barcrawl, higgs, sdss) contain duplicate tuples and we pre-process them by deduplication.}
	\setlength{\tabcolsep}{3pt}
	\small
	\begin{tabular}{c|ccc|cc|ccccc}
		\hline
		Data set & \textit{n} & \textit{m} & UCCs & HPIValid & DUD & Breakdown($\%$) & IOtimes & $|S_{FG}|$ & $P_{pru}$ \\
		\hline
		
		reflns & 24,769 & 37 & 17 & 0.26 & 0.65 & 
		\begin{tikzpicture}
			\draw[gray, fill=yellow] (0,0) rectangle (1.986644407345576, 0.2);
			\draw[gray, fill=green] (1.986644407345576, 0.0) rectangle (1.993322203672788, 0.2);
			\draw[gray, fill=red] (1.993322203672788, 0.0) rectangle (2, 0.2);
		\end{tikzpicture} & 16.89 & 4 & 1.0 \\
		
		entytysrcgen & 26,139 & 46 & 3 & 0.25 & 1.29 & 
		\begin{tikzpicture}
			\draw[gray, fill=yellow] (0,0) rectangle (1.9955752212389382, 0.2);
			\draw[gray, fill=green] (1.9955752212389382, 0.0) rectangle (1.997787610619469, 0.2);
			\draw[gray, fill=red] (1.997787610619469, 0.0) rectangle (2, 0.2);
		\end{tikzpicture} & 13.91 & 3 & 1.0\\ 
		
		atom & 160,000 & 31 & 691 & 16.61 & \textbf{5.73} & 
		\begin{tikzpicture}
			\draw[gray, fill=yellow] (0,0) rectangle (0.5832359813084111, 0.2);
			\draw[gray, fill=green] (0.5832359813084111, 0.0) rectangle (0.5908294392523364, 0.2);
			\draw[gray, fill=red] (0.5908294392523364, 0.0) rectangle (2, 0.2);
		\end{tikzpicture} & 31.87 & 5 & 0.712\\ 
		
		sg$\_$bioentry & 184,292 & 9 & 3 & 0.31 & 0.59 & 
		\begin{tikzpicture}
			\draw[gray, fill=yellow] (0,0) rectangle (1.9898305084745762, 0.2);
			\draw[gray, fill=green] (1.9898305084745762, 0.0) rectangle (1.9932203389830507, 0.2);
			\draw[gray, fill=red] (1.9932203389830507, 0.0) rectangle (2, 0.2);
		\end{tikzpicture} & 16.6 & 4 & 1.0 \\
		
		rangequeriesaggr & 200,000 & 8 & 19 & 0.66 & 0.84 & 		\begin{tikzpicture}
			\draw[gray, fill=yellow] (0,0) rectangle (1.9911894273127753, 0.2);
			\draw[gray, fill=green] (1.9911894273127753, 0.0) rectangle (1.9955947136563876, 0.2);
			\draw[gray, fill=red] (1.9955947136563876, 0.0) rectangle (2, 0.2);
		\end{tikzpicture} & 28.55 & 8 & 1.0\\		
		
		nba2023shots & 217,207 & 26 & 64 & 3.48 & 14.33 & 
		\begin{tikzpicture}
			\draw[gray, fill=yellow] (0,0) rectangle (1.8100340259176138, 0.2);
			\draw[gray, fill=green] (1.8100340259176138, 0.0) rectangle (1.8107579816115253, 0.2);
			\draw[gray, fill=red] (1.8107579816115253, 0.0) rectangle (2, 0.2);
		\end{tikzpicture} & 16.85 & 3 & 0.388\\ 		
		
		sgemm$\_$product & 241,600 & 18	& 36 & 2.12 & \textbf{2.07} & 
		\begin{tikzpicture}
			\draw[gray, fill=yellow] (0,0) rectangle (1.8892760356174991, 0.2);
			\draw[gray, fill=green] (1.8892760356174991, 0.0) rectangle (1.8931475029036005, 0.2);
			\draw[gray, fill=red] (1.8931475029036005, 0.0) rectangle (2, 0.2);
		\end{tikzpicture} & 17.78 & 4 & 0.972\\		
		
		prsa & 420,768 & 18 & 378 & 22.09 & 38.64 & 
		\begin{tikzpicture}
			\draw[gray, fill=yellow] (0,0) rectangle (0.09729811396771837, 0.2);
			\draw[gray, fill=green] (0.09729811396771837, 0.0) rectangle (0.10097449197995617, 0.2);
			\draw[gray, fill=red] (0.10097449197995617, 0.0) rectangle (2, 0.2);
		\end{tikzpicture} & 23.78 & 1 & 0.001\\
		
		onlineretail & 536,641 & 8 & 1 & 0.69 & 1.68 & 
		\begin{tikzpicture}
			\draw[gray, fill=yellow] (0,0) rectangle (1.9977426636568851, 0.2);
			\draw[gray, fill=green] (1.9977426636568851, 0.0) rectangle (1.9988713318284426, 0.2);
			\draw[gray, fill=red] (1.9988713318284426, 0.0) rectangle (2, 0.2);
		\end{tikzpicture} & 10.55 & 2 & 1.0\\ 
		
		struct$\_$sheet$\_$range & 664,128 & 32 & 167 & 16.72 & \textbf{13.63} & \begin{tikzpicture}
			\draw[gray, fill=yellow] (0,0) rectangle (1.997393853838636, 0.2);
			\draw[gray, fill=green] (1.997393853838636, 0.0) rectangle (1.9989141057660982, 0.2);
			\draw[gray, fill=red] (1.9989141057660982, 0.0) rectangle (2, 0.2);
		\end{tikzpicture} & 49.87 & 15 & 1.0 \\
		
		ht$\_$sensor & 928,991 & 12 & 147 & 14.674 & \textbf{8.96} & 
		\begin{tikzpicture}
			\draw[gray, fill=yellow] (0,0) rectangle (1.9956906870960023, 0.2);
			\draw[gray, fill=green] (1.9956906870960023, 0.0) rectangle (1.9990423749102229, 0.2);
			\draw[gray, fill=red] (1.9990423749102229, 0.0) rectangle (2, 0.2);
		\end{tikzpicture} & 37.69 & 11 & 1.0 \\ 		
		
		foursquare$\_$spots & 957,462 & 15 & 1 & 2.03 & 3.51 & 
		\begin{tikzpicture}
			\draw[gray, fill=yellow] (0,0) rectangle (1.9978575254418853, 0.2);
			\draw[gray, fill=green] (1.9978575254418853, 0.0) rectangle (1.998393144081414, 0.2);
			\draw[gray, fill=red] (1.998393144081414, 0.0) rectangle (2, 0.2);
		\end{tikzpicture} & 16.75 & 4 & 1.0 \\ 		
		
		har70plus & 2,259,597 & 8 & 1 & 2.49 & \textbf{2.37} & 
		\begin{tikzpicture}
			\draw[gray, fill=yellow] (0,0) rectangle (1.9977502812148484, 0.2);
			\draw[gray, fill=green] (1.9977502812148484, 0.0) rectangle (1.998875140607424, 0.2);
			\draw[gray, fill=red] (1.998875140607424, 0.0) rectangle (2, 0.2);
		\end{tikzpicture} & 7.55 & 1 & 1.0 \\ 
		
		bitcoinheist & 2,916,697 & 10 & 2 & 5.67 & \textbf{4.51} & 
		\begin{tikzpicture}
			\draw[gray, fill=yellow] (0,0) rectangle (1.998650472334683, 0.2);
			\draw[gray, fill=green] (1.998650472334683, 0.0) rectangle (1.9993252361673415, 0.2);
			\draw[gray, fill=red] (1.9993252361673415, 0.0) rectangle (2, 0.2);
		\end{tikzpicture} & 7.63 & 1 & 1.0 \\ 		
		
		pamap2 & 3,850,505 & 54 & 1206 & 838.83 & \textbf{704.21} & 
		\begin{tikzpicture}
			\draw[gray, fill=yellow] (0,0) rectangle (1.9998617107247654, 0.2);
			\draw[gray, fill=green] (1.9998617107247654, 0.0) rectangle (1.9999942379468654, 0.2);
			\draw[gray, fill=red] (1.9999942379468654, 0.0) rectangle (2, 0.2);
		\end{tikzpicture} & 115.92 & 37 & 1.0 \\ 
		
		ditag$\_$feature & 3,960,124 & 13 & 3 & 9.72 & 28.95 & \begin{tikzpicture}
			\draw[gray, fill=yellow] (0,0) rectangle (1.999811504508184, 0.2);
			\draw[gray, fill=green] (1.999811504508184, 0.0) rectangle (1.9998743363387892, 0.2);
			\draw[gray, fill=red] (1.9998743363387892, 0.0) rectangle (2, 0.2);
		\end{tikzpicture} & 16.71 & 4 & 1.0 \\ 
		
		ssdp$\_$flood & 4,077,266 & 115 & 392 & OM & \textbf{1956.1} & \begin{tikzpicture}
			\draw[gray, fill=yellow] (0,0) rectangle (1.9999790339142163, 0.2);
			\draw[gray, fill=green] (1.9999790339142163, 0.0) rectangle (1.9999979033914217, 0.2);
			\draw[gray, fill=red] (1.9999979033914217, 0.0) rectangle (2, 0.2);
		\end{tikzpicture} & 169.96 & 55 & 1.0 \\
		
		susy & 5,000,000 & 19 & 746 & 294.54 & \textbf{148.76} & 
		\begin{tikzpicture}
			\draw[gray, fill=yellow] (0,0) rectangle (1.9994515940333433, 0.2);
			\draw[gray, fill=green] (1.9994515940333433, 0.0) rectangle (1.999951252802964, 0.2);
			\draw[gray, fill=red] (1.999951252802964, 0.0) rectangle (2, 0.2);
		\end{tikzpicture} & 55.8 & 17 & 1.0 \\		
		
		sdss & 6,949,386 & 8 & 26 & 40.88 & 54.41 & 
		\begin{tikzpicture}
			\draw[gray, fill=yellow] (0,0) rectangle (1.9997139383537155, 0.2);
			\draw[gray, fill=green] (1.9997139383537155, 0.0) rectangle (1.9998569691768577, 0.2);
			\draw[gray, fill=red] (1.9998569691768577, 0.0) rectangle (2, 0.2);
		\end{tikzpicture} & 28.55 & 8 & 1.0 \\						
		
		ids2018 & 7,948,746 & 84 & 48506 & OM & \textbf{1094.5} & 
		\begin{tikzpicture}
			\draw[gray, fill=yellow] (0,0) rectangle (1.998450178090522, 0.2);
			\draw[gray, fill=green] (1.998450178090522, 0.0) rectangle (1.999895751889497, 0.2);
			\draw[gray, fill=red] (1.999895751889497, 0.0) rectangle (2, 0.2);
		\end{tikzpicture} & 10.95 & 2 & 1.0 \\ 
		
		nycparkticket2015 & 9,711,094 & 40 & 1 & 51.98 & 265.46 & \begin{tikzpicture}
			\draw[gray, fill=yellow] (0,0) rectangle (1.999973255357843, 0.2);
			\draw[gray, fill=green] (1.999973255357843, 0.0) rectangle (1.999991085119281, 0.2);
			\draw[gray, fill=red] (1.999991085119281, 0.0) rectangle (2, 0.2);
		\end{tikzpicture} & 10.91 & 2 & 1.0 \\ 
		
		hepmass & 10,500,000 & 29 & 1889 & 1030.78 & \textbf{682.87} & 
		\begin{tikzpicture}
			\draw[gray, fill=yellow] (0,0) rectangle (1.9994540606394964, 0.2);
			\draw[gray, fill=green] (1.9994540606394964, 0.0) rectangle (1.9999856331747237, 0.2);
			\draw[gray, fill=red] (1.9999856331747237, 0.0) rectangle (2, 0.2);
		\end{tikzpicture} & 67.86 & 21 & 1.0 \\ 
		
		higgs & 10,721,302 & 29 & 8121 & 11172.91 & \textbf{3019.6} & 
		\begin{tikzpicture}
			\draw[gray, fill=yellow] (0,0) rectangle (0.2576965844842124, 0.2);
			\draw[gray, fill=green] (0.2576965844842124, 0.0) rectangle (0.2594905113363411, 0.2);
			\draw[gray, fill=red] (0.2594905113363411, 0.0) rectangle (2, 0.2);
		\end{tikzpicture} & 55.86 & 8 & 0.707 \\ 
		
		phaccelerometer & 13,062,475 & 10 & 12 & 39.64 & 67.55 & \begin{tikzpicture}
			\draw[gray, fill=yellow] (0,0) rectangle (1.9998581962563813, 0.2);
			\draw[gray, fill=green] (1.9998581962563813, 0.0) rectangle (1.9998581962563813, 0.2);
			\draw[gray, fill=red] (1.9998581962563813, 0.0) rectangle (2, 0.2);
		\end{tikzpicture} & 13.63 & 3 & 1.0 \\ 
		
		barcrawl & 14,057,535 & 5 & 1 & 19.32 & \textbf{12.37} & 
		\begin{tikzpicture}
			\draw[gray, fill=yellow] (0,0) rectangle (1.9996853366897422, 0.2);
			\draw[gray, fill=green] (1.9996853366897422, 0.0) rectangle (1.9998426683448711, 0.2);
			\draw[gray, fill=red] (1.9998426683448711, 0.0) rectangle (2, 0.2);
		\end{tikzpicture} & 7.33 & 1 & 1.0\\ 
		
		hhar & 26,995,107 & 10 & 7 & 84.78 & 165.71 & \begin{tikzpicture}
			\draw[gray, fill=yellow] (0,0) rectangle (1.9999598444642246, 0.2);
			\draw[gray, fill=green] (1.9999598444642246, 0.0) rectangle (1.9999732296428163, 0.2);
			\draw[gray, fill=red] (1.9999732296428163, 0.0) rectangle (2, 0.2);
		\end{tikzpicture} & 13.63 & 3 & 1.0\\ 
		
		ghtorrent & 73,749,949 & 11 & 2 & OM & \textbf{646.77} & 
		\begin{tikzpicture}
			\draw[gray, fill=yellow] (0,0) rectangle (1.496470151676794, 0.2);
			\draw[gray, fill=green] (1.496470151676794, 0.0) rectangle (1.4964732439661705, 0.2);
			\draw[gray, fill=red] (1.4964732439661705, 0.0) rectangle (2, 0.2);
		\end{tikzpicture} & 10.67 & 1 & 0.5\\ 
		
		hi-large$\_$trans & 179,702,079 & 11 & 3 & OM & \textbf{6447.6} & \begin{tikzpicture}
			\draw[gray, fill=yellow] (0,0) rectangle (1.9999972665779013, 0.2);
			\draw[gray, fill=green] (1.9999972665779013, 0.0) rectangle (1.9999978740050344, 0.2);
			\draw[gray, fill=red] (1.9999978740050344, 0.0) rectangle (2, 0.2);
		\end{tikzpicture} & 19.66 & 3 & 1.0 \\ 
		
		\hline
	\end{tabular}
	\label{table:realtableresults}
\end{table*}

\subsection{Experiment 3: the real-life data sets}

Experiment 3 evaluates the performance of DUD on real-life data sets, which range from small scale (24K tuples) to large scale (179M tuples), from low dimension (5) to high dimension (115), representing the practical data characteristics comprehensively. The execution results of the algorithms are reported in Table \ref{table:realtableresults}, where \textit{IOtimes} is the ratio of the retrieved byte number of DUD over original table size, \textit{breakdown} shows the time proportions of three phases for DUD, $|S_{FG}|$ is the number of attributes of FG type.

On the real-life data sets whose tuple numbers are less than $10^6$, HPIValid only needs to retrieve data sets once, while DUD often involves an order of magnitude higher I/O cost. The I/O cost of a sorting operation is at least twice the size of the relation. When data scale is small and the stripped partitions fit in memory, HPIValid benefits from both lower I/O cost and efficient in-memory validation, and tends to run faster than DUD. DUD performs better when the number of attributes is large, as seen on atom ($m = 31$) and struct\_sheet\_range ($m = 32$).

On some larger real-life data sets where the PLI-based structures of HPIValid fit within the allocated memory, HPIValid tends to outperform DUD, as seen on nycparkticket2015, hhar, phaccelerometer, foursquare\_spots, ditag\_feature and sdss. On these data sets, DUD incurs the full cost of disk-based sorting and scanning, while HPIValid validates candidates directly from its in-memory structures. 

On pamap2 ($m = 54$), hepmass ($m = 29$), higgs ($m = 29$), susy ($m = 19$) and ht\_sensor ($m = 12$), HPIValid still completes but DUD runs faster. This is mainly because DUD prunes a large proportion of candidates on these data sets; the reduction in validation cost outweighs the cost of disk-based sorting and scanning. When the PLI-based structures exceed the memory budget, HPIValid runs out of memory while DUD still completes. This is the case for ssdp\_flood ($m = 115$), ids2018 ($m = 84$), ghtorrent ($n = 73.7$M) and hi-large\_trans ($n = 179.7$M). DUD is not uniformly faster than HPIValid. Its advantage is most significant on data sets where PLI-based structures are too large to fit in memory.

The last column $P_{pru}$ in Table \ref{table:realtableresults} measures the proportion of the UCC candidates which are reported as true UCCs without validation. Obviously, all generated UCC candidates can be declared as true UCCs in the very great majority of cases. A notable exception is the prsa data set , which only reports one UCC candidate as true UCC without validation. On prsa, the first attribute has much greater cardinality (35064) than the maximum cardinality (2028) of other attributes. Thus $cmp_1$ on prsa is much smaller than $cmp_i (2 \le i \le 18)$. Given $\gamma$ value, only the first attribute is of FG type. For prsa, $\{ A_1, A_{18} \}$ is the only UCC containing $A_1$, and other UCCs have to be validated since they do not include $A_1$. This also explains the so large proportion of phase 3 for prsa execution. 

In a few rare cases, including three real-life data sets (NBA2023Shots, IDS2018 and HI-Large$\_$Trans), we have $\min\limits_{1 \le i \le m} cmp_i > 3 \times n \times \log_2 n$, there is no attribute of FG type if $\gamma$ is set to $3 \times n \times \log_2 n$, and $\gamma$ is set to $5 \times \min\limits_{1 \le i \le m} cmp_i$. In a word, parameter $\gamma$ is determined for DUD to guarantee at least one attribute of FG type with a relative low cost. It should be noted that, all factors for $n \times \log_2 n$ in the three data sets are greater than 3 (default factor) and are no more than 30. The computed $\gamma$ value is still in the same magnitude of the default $\gamma$ value ($3 \times n\times \log_2 n$).

\begin{figure*}
	\centering
	\renewcommand{\thesubfigure}{}
	\subfigure[(a) Execution time]{
		\includegraphics[scale = 0.63]{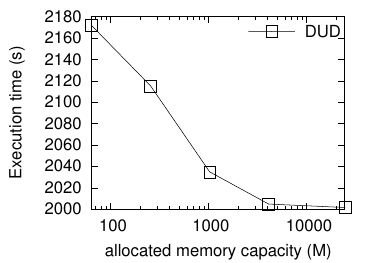}}
	\subfigure[(b) Memory requirement]{
		\includegraphics[scale = 0.63]{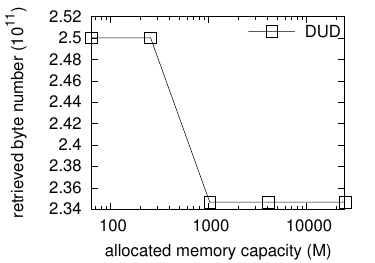}}
	\subfigure[(c) Memory requirement]{
		\includegraphics[scale = 0.63]{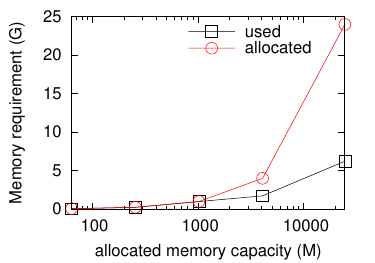}}
	\caption{The illustration of different memory capacity.}
	\label{fig:memory}
\end{figure*}

\subsection{Experiment 4: effect of different memory capacity}

Given SF = 10 and $m = 16$, experiment 4 evaluates the performance of DUD under different allocated memory capacity, from 64MB to 24GB. The buffer size used for sorting is set to one fourth of the allocated memory capacity (i.e., 6GB when the JVM heap is 24GB). Files within this size are loaded into memory, sorted, and written back to disk without intermediate temporary files, while larger files are sorted using TPMMS. This also verifies the usability of DUD under very small memory allocation. 

As illustrated in Figure \ref{fig:memory}(a), the execution time of DUD increases gradually when the allocated memory capacity decreases from 24GB to 64MB. It should be noted that the execution time under different memory capacity actually varies within a small range. This shows the advantages of DUD that it can run on a very small memory capacity and its performance does not seriously degenerate. The gradual increase of the execution time in experiment 4 with smaller memory capacity can be explained by more I/O cost. A smaller memory capacity makes sort operation generate a larger number of sorted partitions and change from in-memory sorting to two-pass multi-way merge sort when file to sort cannot fit in the sorting buffer. The latter incurs a higher I/O cost, as illustrated in Figure \ref{fig:memory}(b). The memory requirement of DUD is illustrated in Figure \ref{fig:memory}(c), the actual memory consumption for DUD always lies in the range the allocated maximum capacity, and DUD can discover UCCs on disk-resident data with a very small memory capacity.

\subsection{Discussion} \label{sec:exp:discussion}

The experimental evaluation in this section helps us understand DUD more comprehensively and clearly. Evidently, the biggest strength for DUD is its ability to deal with large-scale disk-resident data under limited memory, while the existing algorithms cannot achieve this goal. This reflects a trade-off between memory and I/O: DUD avoids maintaining global in-memory structures such as stripped partitions, at the cost of additional disk-based sorting and sequential scans. When PLI-based structures of HPIValid fit within the memory budget, the two algorithms are roughly comparable, with each winning on roughly half the data sets, as shown in Experiment 3. This also validates the beneficial effect of candidate pruning and hash-based batch validation. For DUD, the ideal case is that all UCC candidates contain at least one attribute of FG type. In this case, all candidates are true UCCs definitely and do not need to be validated. The execution of DUD consists of three phases, and the execution time of phase 2 can be negligible. The performance of DUD actually is a trade-off between phase 1 and phase 3, which is adjusted by the different factor selection of $\gamma$ value. The smaller $\gamma$ value leads to a fast execution of phase 1 but maybe a long execution of phase 3. The greater $\gamma$ value leads to the opposite case. The default factor 3 in this paper is empirically shown to be the optimal choice, as described in Section \ref{sec:exp:factorselection}.

The memory consumption of DUD during validation depends on the maximum run length of the NG attribute used for validation, namely the largest number of tuples sharing the same value on that attribute. This quantity is data-dependent and cannot be bounded by a fixed fraction of the relation size in general. For each candidate, DUD uses the NG attribute with the smallest maximum run length among its constituent NG attributes for validation, since NG attributes are visited in ascending order of maximum run length. A nearly constant attribute is therefore not used when another NG attribute of the same candidate yields smaller runs. A pathological case may still arise when all NG attributes in a candidate have large maximum run lengths. In that case, a validation batch can be large, and DUD does not provide a formal worst-case memory bound. In our experiments, the number of tuples maintained in memory during validation is several orders of magnitude smaller than the total tuple count: 50,554 out of 119,994,608 tuples on lineitem SF $= 20$ (Figure \ref{fig:tupnum}(f)).

DUD also has limitations. When the number of attributes is large or the UCC result set is very large, the costs of candidate management and validation may increase substantially, particularly when few attributes are of FG type and the pruning effect is limited. In such cases, if the PLI-based structures of HPIValid still fit within the memory budget, HPIValid may run faster than DUD because it can benefit from in-memory PLI-based validation, whereas DUD still pays the cost of disk-based sorting and sequential scans. Therefore, we do not claim that DUD dominates in-memory PLI-based algorithms on all high-dimensional or large-result-set data sets. Its main advantage lies in memory-constrained settings where maintaining global PLI-based structures becomes difficult, and the pruning by Theorem \ref{theorem:candidatepruning} further contributes when it can substantially reduce validation cost. On data sets with highly imbalanced attribute cardinalities, such as prsa, only a small number of attributes may be classified as FG type, in which case phase 3 validation dominates the running time.

Among the existing UCC discovery algorithms discussed in Section \ref{sec:relatedworks}, GORDIAN, DUCC, and HyUCC mainly rely on global in-memory data structures. Under the single-node memory budget considered in this paper, these algorithms cannot be directly applied once their required structures exceed the available memory. We do not make a direct experimental comparison with these algorithms in this paper. Instead, HPIValid is used as the main baseline because it is the most recent state-of-the-art UCC discovery algorithm and has been reported to substantially outperform these earlier methods \cite{DBLP:journals/pvldb/BirnickBFNPS20}. The comparison here is therefore qualitative. When the required structures fit within the memory budget, these algorithms can avoid some of the disk-based sorting and sequential scans required by DUD. When such structures exceed the available memory, their applicability becomes limited, whereas DUD remains applicable. DUD is most suitable for settings where global in-memory structures cannot be held in memory, and where the pruning by Theorem \ref{theorem:candidatepruning} eliminates a large fraction of candidates from validation.

The disk-resident design of DUD has direct relevance to large-scale data profiling in practice. Data sets in data quality, schema discovery, and data integration tasks are often stored on disk, and available memory may be far smaller than the data size. In such settings, algorithms that require global in-memory structures cannot run. DUD shows that UCC discovery remains feasible under these constraints, at the cost of additional disk I/O. In our experiments, DUD completes on data sets with up to 179.7 million tuples where HPIValid runs out of memory. 

\textit{Open problems}. Several directions remain for future work. A validation strategy with a provable memory bound would address the data-dependent memory behavior discussed above. Improving performance on high-dimensional data sets with very large UCC result sets is also an open problem. Replacing the fixed $\gamma$ factor used in this paper with a strategy that adapts to data distributions is another possibility. Extending DUD to a distributed setting, to scale beyond a single machine, is a further direction. Pipelining sorting and validation is another direction for reducing I/O cost in future implementations, although it would require additional design to stream sorted tuples into validation while preserving the reuse of sorted results across validation batches. A more refined pruning strategy, replacing the binary FG/NG classification with attribute-level ranking or partial difference set generation for NG attributes, could mitigate the validation cost on data sets with highly imbalanced attribute cardinalities such as prsa.

\section{Conclusion} \label{sec:conclusion}

This paper considers the problem of discovering unique column combinations (UCCs) on disk-resident data with limited memory. The existing algorithms need to maintain in-memory data structures for full data, and cannot deal with large-scale data well. This paper proposes a novel DUD algorithm to compute UCCs on large-scale disk-resident data efficiently. The execution of DUD consists of three phases. In phase 1, DUD determines the types of attributes with a low cost, generates full useful difference sets with respect to the attributes of FG type, and supplements difference sets by random sampling. Then DUD constructs the hypergraph, whose vertices are attributes of relation schema and hyperedges correspond to the generated difference sets. The UCC candidates can be generated quickly by enumerating minimal hitting sets of constructed hypergraph in phase 2. An important theorem is proved that the candidates including at least one attribute of FG type are true UCCs without validation, which significantly reduces the number of candidates to be validated in phase 3. For the validation candidates, phase 3 exploits hash-based batch validation strategy. A small number of tuples are maintained at a time to validate a set of candidates simultaneously. Any invalid candidates indicate conflicting tuples, which have equal projections on the candidates. The pairwise comparisons are performed for the conflicting tuples to generate new hyperedges and update the hypergraph. Then, if new hyperedges exist, DUD enters phase 2 to find the new UCC candidates by the similar follow-up processing, with the only difference that the UCC candidates already discovered to be UCCs in the previous operations should be removed. The extensive experimental results, conducted on synthetic and real-life data sets, show that DUD can discover UCCs on disk-resident data with limited memory quickly.

\bibliographystyle{IEEEtran}
\bibliography{references}


\vfill

\end{document}